\documentclass[lettersize,journal]{IEEEtran}
\usepackage{amsfonts}
\usepackage{amssymb}
\usepackage{amsthm}
\usepackage{amsmath,amsfonts,amssymb}
\usepackage{mathrsfs}
\ifCLASSINFOpdf
\usepackage[pdftex]{graphicx}
\else
\fi
\ifCLASSOPTIONcompsoc
\usepackage[caption=false,font=normalsize,labelfont=sf,textfont=sf]{subfig}
\else
\usepackage[caption=false,font=footnotesize]{subfig}
\fi

\usepackage{verbatim}
\usepackage{setspace}
\usepackage{bm}
\usepackage{bbm}
\usepackage{algorithm}
\usepackage{algorithmic}
\ifCLASSOPTIONcompsoc
\else
\usepackage[caption=false,font=footnotesize]{subfig}
\fi
\usepackage{cite}
\usepackage{color}
\usepackage{multirow}
\usepackage{setspace}

\usepackage{changepage}
\usepackage{pdfpages}
\usepackage{color}

\newtheorem{theorem}{Theorem}
\newtheorem{lemma}{Lemma}

\usepackage{flushend} %最后一页，底部对齐

\usepackage{gensymb}
\usepackage{booktabs}
\usepackage{threeparttable}
\usepackage{tabularx}
\usepackage{multirow}
\usepackage{array}
\newcommand{\PreserveBackslash}[1]{\let\temp=\\#1\let\\=\temp}
\newcolumntype{C}[1]{>{\PreserveBackslash\centering}p{#1}}
\newcolumntype{R}[1]{>{\PreserveBackslash\raggedleft}p{#1}}
\newcolumntype{L}[1]{>{\PreserveBackslash\raggedright}p{#1}}

\usepackage{url}
\makeatletter
\renewcommand{\maketag@@@}[1]{\hbox{\m@th\normalsize\normalfont#1}}%
\makeatother

\IEEEoverridecommandlockouts
\begin{document}
% 激活bst控制：主要包括超过6个以上的作者使用et al.等控制
\bstctlcite{IEEEexample:BSTcontrol} 
\title{Movable Antenna Empowered Secure Near-Field MIMO Communications}
\author{Yaodong~Ma,~\IEEEmembership{Graduate Student Member,~IEEE,}
	Kai~Liu,~\IEEEmembership{Member,~IEEE,}\\
	Yanming~Liu,~\IEEEmembership{Member,~IEEE,}
	and Lipeng~Zhu,~\IEEEmembership{Member,~IEEE}
	% <-this % stops a space
	\thanks{This work is supported by the National Natural Science Foundation of
		China under Grant No. U2133210 and U2233216, and the National Key
		R\&D Program of China under Grant No. 2023YFB4302800. \textit{(Corresponding author: Lipeng Zhu)}
		
		Y. Ma and K. Liu are with the School of Electronics and Information Engineering, Beihang University, Beijing, 100191, China, and the State Key Laboratory of CNS/ATM, Beijing, 100191, China (e-mail: \{yaodongma, liuk\}@buaa.edu.cn).
		
		Y. Liu is with the China Satellite Network System Company Ltd., Beijing, 100071, China (e-mail: liuyanming@buaa.edu.cn).
		
		L. Zhu is with the Department of Electrical and Computer Engineering, National
		University of Singapore, Singapore 117583 (e-mail: zhulp@nus.edu.sg).
	}
	% <-this % stops a space
	\thanks{}}
\vspace{-10mm}
\maketitle
\vspace{-10mm}

\begin{abstract}
This paper investigates movable antenna (MA) empowered secure transmission in near-field multiple-input multiple-output (MIMO) communication systems, where the base station (BS) equipped with an MA array transmits confidential information to a legitimate user under the threat of a potential eavesdropper. To enhance physical layer security (PLS) of the considered system, we aim to maximize the secrecy rate by jointly designing the hybrid digital and analog beamformers, as well as the positions of MAs at the BS. To solve the formulated non-convex problem with highly coupled variables, an alternating optimization (AO)-based algorithm is introduced by decoupling the original problem into two separate subproblems. Specifically, for the subproblem of designing hybrid beamformers, a semi-closed-form solution for the fully-digital beamformer is first derived by a weighted minimum mean-square error (WMMSE)-based algorithm. Subsequently, the digital and analog beamformers are determined by approximating the fully-digital beamformer through the manifold optimization (MO) technique. For the MA positions design subproblem, we utilize the majorization-minimization (MM) algorithm to iteratively optimize each MA’s position while keeping others fixed. Extensive simulation results validate the considerable benefits of the proposed MA-aided near-field beam focusing approach in enhancing security performance compared to the traditional far-field and/or the fixed position antenna (FPA)-based systems. In addition, the proposed scheme can realize secure transmission even if the eavesdropper is located in the same direction as the user and closer to the BS.
\end{abstract} %220words


% \fi
\IEEEpeerreviewmaketitle

\section{Introduction}
\IEEEPARstart{T}{he} explosive growth of wireless communications has not only expanded network coverage and increased heterogeneity but also raised significant concerns regarding security and privacy \cite{nguyen2021security}. The inherent openness of wireless communications exposes transmitted data to potential interception by unauthorized entities. To tackle this challenge, physical layer security (PLS) technologies offer the advantage of low computational complexity and low probability of interception \cite{wang2018survey}, making them a promising alternative to conventional encryption/decryption methods. 
To enhance the secrecy performance of PLS, several technologies have been proposed, including cooperative jamming \cite{liao2024jamming}, artificial noise (AN) \cite{ding2024secure}, and secure beamforming \cite{mukherjee2021secrecy, shi2015secure, song2024deep}, among which beamforming has emerged as an effective technique for enhancing PLS by leveraging spatial degrees of freedom (DoFs).
%in the fifth-generation (5G) and pre-5G communication systems, with the advancement of multi-antenna technology and the receiver always being in the far-field region of the transmitter, beamforming has emerged as an effective technique for enhancing PLS by leveraging spatial degrees of freedom (DoFs) \cite{mukherjee2021secrecy, shi2015secure, song2024deep}.
This method directs beams toward legitimate users while minimizing signal leakage toward potential threats through the coordinated transmission of multiple antennas. Besides, the secrecy performance relies on the size of the antenna array \cite{liu2023near}. Specifically, as the array size increases, the main lobe of the beam becomes narrower, enabling more directional transmissions with higher beamforming gain to legitimate receivers and lower power leakage to eavesdroppers.
%, thus enhancing secrecy performance.
Therefore, it is a significant trend to explore the potential advantages of secure communications by using large-aperture antenna arrays. Nevertheless, deploying large-scale antenna arrays in high-frequency bands results in a significantly increase in the Rayleigh distance. As a result, transceivers may fall within the near-field regions of each other, causing transition from far-field beamforming to near-field beamforming (also known as beam focusing \cite{zhang2022beam, cui2022near}). This indicates that PLS designs based on the far-field channel, which assumes plane wavefronts of the electromagnetic field rather than spherical ones, are no longer applicable in near-field applications. 
%However, as the array aperture and carrier frequencies grow, wireless communications inevitably shift into the near-field region \cite{liu2023near}. This indicates PLS designs based on the far-field model, which assumes plane wavefronts of the electromagnetic (EM) field rather than spherical ones, no longer holds in near-field applications. Consequently, specialized designs are required to effectively enhance the PLS in near-field systems.

The essential difference between near-field and far-field beamforming lies in the fact that the former can focus energy in both the angular and distance domains, whereas far-field beamforming is restricted to steering energy only in the angular domain. 
%Consequently, near-field beamforming is also referred to as near-field beam focusing \cite{zhang2022beam, cui2022near}, while far-field beamforming is commonly called far-field beam steering. 
The near-field channel model introduces an additional spatial dimension, which offers extra DoFs over the far-field model, thereby improving PLS transmission performance. Nevertheless, most studies on PLS have primarily assumed far-field channel conditions and secure communications via beam steering \cite{maY2024movable,tang2024secure,dong2020enhancing}. This assumption limits the potential improvement in secrecy performance provided by spatial beamforming in near-field scenarios. For example, if the user and eavesdropper are situated in the same direction relative to the base station (BS), the angular domain beam steering vector cannot differentiate between them, resulting in confidential information being intercepted by the eavesdropper. To explore the potential of near-field secure communications, several studies \cite{zhao2024performance, zhang2024near, zhang2024physical} have investigated near-field beam focusing in terms of PLS. Specifically, the authors in \cite{zhao2024performance} analyzed the secrecy performance under near-field and far-field conditions and obtained the secrecy rate in closed form. The authors in \cite{zhang2024near} maximized the secrecy rates in near-field wideband communication systems by jointly optimizing the analog beamformer and transmit power. Besides, the authors in \cite{zhang2024physical} developed a two-stage algorithm to maximize the secrecy rate of near-field multiple-input multiple-output (MIMO) communication systems by utilizing hybrid beamformers. 
However, the mentioned traditional fixed-position antenna (FPA)-based MIMO systems require the deployment of an extremely large number of antennas, along with their associated radio frequency (RF) front ends at the BS. For example, in a square region of size $100$ wavelengths, the total number of antennas in an array with half-wavelength inter-antenna spacing is $4 \times 10^4$. Such high hardware costs, energy consumption, and signal processing overhead may hinder the efficiency of secure communications in near-field scenarios.
%However, the above works mainly focused on secure transmission design in near-field scenarios using fixed-position antennas (FPAs), while the channel spatial variation was constrained therein, leading to performance limitations.

Recently, to overcome the performance limitations of conventional FPAs, the concept of the movable antenna (MA), also referred to the fluid antenna system, has been introduced to wireless systems \cite{zhu2023movableMag, zhu2025tutorial, zhu2024historical}.
%enabling flexible adjustment of antenna element positions at the transceiver to enhance communication performance \cite{zhu2023modeling, ma2023mimo}. 
MA is a cost-effective solution that requires fewer antennas and RF chains, and equivalently achieves beam focusing by simply expanding the antenna moving region, demonstrating the significant potential for secure communications \cite{zhu2023modeling, ma2023mimo}.
Various works have validated the benefits of MA-aided communication systems in terms of spatial diversity gain \cite{zhu2023modeling, zhu2024performance,mei2024movable,xiao2024channel}, spatial multiplexing performance \cite{ma2023mimo, shao20246d, shao20256dma, shao20256d}, multiuser interference mitigation \cite{zhu2023movable, xiao2023multiuser, zhou2024movable}, and flexible
beamforming \cite{zhu2023movable, lyu2024flexible, zhu2024dynamic, ma2024movable}. 
%In addition, the applications of MA arrays in fields such as unmanned aerial vehicle (UAV) communications \cite{tang2024uav}, satellite communications \cite{zhu2024dynamic}, wireless sensing \cite{ma2024movable}, and integrated sensing and communications (ISAC) \cite{lyu2024flexible, maY2024movable} have also been investigated. 
As for MA-enhanced secure communication, the authors in \cite{maY2024movable} investigated the secure communications in integrated sensing and communications (ISAC) systems by designing a two-layer iterative algorithm, where the transmit and receive beamformers, the antenna positions, and the reconfigurable intelligent surface (RIS) reflection coefficients were jointly optimized. The authors in \cite{hu2024secure, xiong2025secure} maximized the secrecy rate by using the projected gradient ascent (PGA) algorithm, where the transmit beamformers and positions of MAs were jointly optimized.
To further improve the secure transmission performance, the authors in \cite{tang2024secure, ding2024secure, liao2024jamming} utilized AN to mislead eavesdropping. By jointly optimizing AN beamformers and MAs’ positions, interference toward illegitimate receivers is enhanced, thereby reducing the risk of confidential information disclosure. In addition, considering the imperfect channel state information (CSI) available to the eavesdropper, the authors in \cite{hu2024movable} minimized the secrecy outage probability by utilizing statistical CSI. Similarly, MA positions were optimized to maximize the secrecy rate \cite{feng2024movable} and minimize transmit powers \cite{cheng2025movable}. Apart from the PLS techniques, MA-aided covert communication, which aims to enhance transmission security by embedding the signal within noise or AN by intentional randomness, thereby reducing the likelihood of being detected by potential eavesdroppers \cite{liu2024movable, wang2024movable, xie2024movable, cheng2024movable}. 
However, all of the aforementioned works \cite{maY2024movable, tang2024secure, ding2024secure, liao2024jamming, hu2024movable, cheng2025movable, feng2024movable, liu2024movable, wang2024movable, xie2024movable, cheng2024movable} on MA-enabled systems focus on securing far-field transmission and do not reveal the performance upper bound in near-field regions.

In practice, in order to improve communication performance, MA-enabled wireless systems typically require larger movement areas (i.e., aperture sizes) to provide higher spatial flexibility than traditional FPA-based systems \cite{zhu2025movable}. As these areas expand and/or carrier frequencies rise (e.g., mmWave and THz bands), the Rayleigh distance also increases, which can render the far-field assumption inaccurate in upcoming wireless systems, thereby making the near-field spherical-wave model necessary. Consequently, several studies have begun exploring the potential benefits that MAs could provide in near-field situations. Specifically, in \cite{zhu2025movable}, the field response channel model was extended from the far-field to near-field propagation environment for MA systems, where the performance of multiuser communications with both digital and analog beamforming strategies was analyzed.
The authors in \cite{ding2024near} minimized the transmit power for near-field multiuser scenarios, where the MA positions and transmit beamformers were jointly designed. The authors in \cite{wang2025antenna} investigated MA-assisted near-field sensing, focusing on estimating the angle and/or distance of a target using a linear MA array. Besides, another recent work \cite{ding2024movable} employed MAs to improve sensing and communication capabilities. In particular, a two-layer random position approach was proposed to jointly optimize the transmit/receive beamformers, sensing signal covariance matrices, the uplink link power allocation, and MA positions at the BS. Furthermore, the authors in \cite{zhu2024suppressing} utilized MAs to mitigate the beam squint effect in wideband multiple-input single-output (MISO) communication systems.
Although works \cite{zhu2025movable,ding2024near,wang2025antenna,ding2024movable,zhu2024suppressing} investigated MA positions designs for communication or sensing, the role of MAs in securing transmission in near-field systems was not unveiled, and the transceiver designs for far-field conditions cannot be directly applied.
Meanwhile, studies on optimizing MA positions in near-field conditions remain in their early stages.

%-------- Motivation & Contribution --------
Motivated by the above discussions, to fully exploit the DoFs in channel reconfiguration provided by the MAs and near-field beam focusing, this paper investigates MAs-assisted secure transmission in near-field MIMO systems, where the multi-MA-enabled BS transmits private information to the user using multiple FPAs, while an eavesdropper with multiple FPAs is located in the same direction as the user to intercept the legitimate information. 
The main contributions of this paper are summarized as follows:
\begin{itemize}
	\item We consider an MA-aided near-field MIMO communication system in the presence of an eavesdropper, aiming at enhancing the PLS. Specifically, an optimization problem is formulated for maximizing the secrecy rate by jointly optimizing the hybrid beamformers and the positions of MAs at the BS, subject to the transmit power constraint, the analog beamformer unit-modulus constraints, the MA moving region constraint, and the minimum inter-MA distance constraint.
	\item We propose an algorithm based on the alternating optimization (AO) method to decompose the complex non-convex optimization problem into two independent subproblems. For the design of hybrid beamforming at the BS, a two-stage approach is developed. During the first stage, the fully-digital beamformer is obtained in semi-closed-form using a weighted minimum mean-square error (WMMSE) algorithm. During the second stage, analog and baseband digital beamformers are derived by employing the manifold optimization (MO) approach to approximate the fully digital beamformers. For optimizing MA positions, we employ the majorization-minimization (MM) technique to iteratively adjust the position of each MA while keeping the others fixed.
	\item Extensive simulation results are presented to validate the considerable benefits of the MAs-assisted schemes for achieving a higher secrecy rate compared to other benchmarks. Furthermore, the proposed approach enables secure transmission even when the eavesdropper is positioned in the same direction as the user and is closer to the BS.
\end{itemize}

The remainder of the paper is organized as follows. Section II introduces the system model and problem formulation. Section III presents the detailed solution for the considered MA-aided near-field MIMO communication systems. Section IV provides the simulation results, and conclusions are finally drawn in Section V.

\textit{Notation}: $a$, $\mathbf{a}$, and $\mathbf{A}$ represent a scalar, a vector, and a matrix, respectively. Besides, for a vector $\mathbf{a}$, $\mathbf{a}_i$ denote its $i$-th element. The 2-norm of vector $\mathbf{a}$ is defined by $\|\mathbf{a}\|$. $|\mathbf{A}|$ and $\|\mathbf{A}\|_F$ denote the determinant and Frobenius norm of matrix $\mathbf{A}$, respectively. $(\cdot)^{\rm T}$, $(\cdot)^{*}$, and $(\cdot)^{\rm H}$ represent the transpose, conjugate, and Hermitian transpose operations, respectively. $\mathbf{I}_N$ indicates the $N$-order identity matrix. In addition, we denote the complex circularly symmetric Gaussian distribution with mean zero and covariance $b$ as $\mathcal{CN}\left(0,b\right)$. ${\rm diag}\left\{\mathbf{a}\right\}$ represents a diagonal matrix with the elements of vector $\mathbf{a}$ on the main diagonal. $\circ$ and $\otimes$ denote the Hadamard product and Kronecker product, respectively.  The gradient vector of the function $f$ with respect to (w.r.t.) vector $\mathbf{x}$ is denoted by $\nabla_{\mathbf{x}}f(\mathbf{x})$, and $\nabla^2_{\mathbf{x}}f(\mathbf{x})$ denotes the Hessian matrix of $f$ w.r.t. vector $\mathbf{x}$. The real and imaginary components of a complex number are represented by $\Re\{\cdot\}$ and $\Im\{\cdot\}$, respectively. The phase of complex number $a$ is represented by $\angle{a}$. The symbol $\mathbb{E}\{\cdot\}$ represents the statistical expectation. Furthermore, ${\rm vec}(\mathbf{V})$ and ${\rm Tr}(\mathbf{V})$ are the vectorization operations of the matrix and its trace, respectively.

\section{System Model and Problem Formulation}
\begin{figure}[t]
	\begin{center}
		\includegraphics[width= 3.3 in]{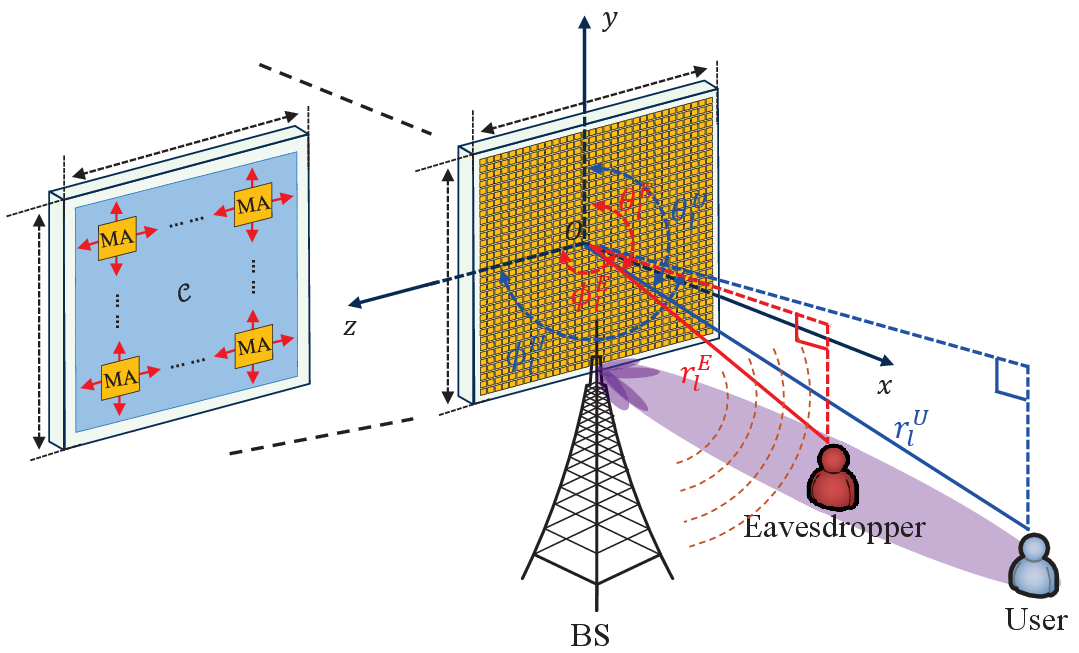}
		\caption{Illustration of the MA-aided near-field MIMO communication system.} \label{fig::systemModel}
	\end{center}
	\vspace{-1.5em}
\end{figure}

In this paper, we consider an MA-aided near-field MIMO communication system, as illustrated in Fig. \ref{fig::systemModel}, where the BS transmits information to the legitimate user, while an eavesdropper attempts to intercept the data. The BS is equipped with an MA array consisting of $M$ antenna elements, while the user and eavesdropper are each equipped with $L_U \geq 2$ and $L_E \geq 2$ FPAs, respectively. To reduce the RF chain overhead, the BS adopts a hybrid beamforming architecture \cite{yu2016alternating, zhang2024physical}, while a fully-digital antenna architecture is employed at both the user and eavesdropper for simplicity.
We assume that the distance from the BS to both the user and eavesdropper is less than the Rayleigh distance, which is expressed as $d_R = 2D^2/\lambda$, where $D$ and $\lambda$ denote the array aperture and the wavelength, respectively. In other words, both the user and eavesdropper are assumed to be positioned within the BS's near-field region. Thus, the transmitted signals from the BS to both the user and eavesdropper follow spherical-wave propagation. Moreover, we address a challenging scenario in which the eavesdropper is situated between the BS and user, with identical azimuth angles. It is worth noting that this scenario cannot guarantee secure transmission through beamforming in far-field conditions. To resist eavesdropping, the near-field transmission directs the beam pattern toward a particular region, i.e., beam focusing \cite{sohrabi2016hybrid}, making the MA array at the BS promising for achieving secure communications.

\subsection{Channel Model}
As illustrated in Fig.~\ref{fig::systemModel}, a three-dimensional (3D) Cartesian coordinate system is established to describe the positions of MAs, where the MA array at the BS resides in the $y$-$O$-$z$ plane with its central point placed at the coordinate origin $(0,0,0)$. Let $\mathbf{t}_m$, $1 \leq m \leq M$, denote the position of the $m$-th MA, which is confined to the antenna moving region, $\mathcal{C} \triangleq \{(0,y,z)|y \in [-A/2, A/2], z \in [-A/2, A/2]\}$. Besides, we define $I \in \{U, E\}$ as the set consisting of the user and the eavesdropper for further operations. Thus, the position of the $l$-th antenna at the user/eavesdropper is given by $\mathbf{r}^{I}_l = [r^{I}_l \cos\theta^{I}_l \sin\phi^{I}_l, r^{I}_l \sin\theta^{I}_l \sin\phi^{I}_l, r^{I}_l \cos\phi^{I}_l]^{\rm T}$, $1 \leq l \leq L_{{I}}$, where $\theta^{I}_l$, $\phi^{I}_l$, and $r^{I}_l$ are denoted as the azimuth angle, elevation angle, and the distance from the origin to the $l$-th antenna at user/eavesdropper, respectively.

Since the line-of-sight (LoS) path plays a crucial role compared to non-LoS paths in typical near-field environments over high-frequency bands \cite{cui2022near}, the LoS channel from the BS to the user is adopted. 
As such, the near-field response vector (NFRV) for the channel from the antenna position $\mathbf{t}_m$ to the $l$-th antenna at the legitimate user is represented by
\begin{equation} \label{eq::NFRV}
	\mathbf{a}(\mathbf{t}_m) = \left[ e^{j\frac{2\pi}{\lambda}\|\mathbf{t}_m-\mathbf{r}^U_1\|_2}, \cdots, e^{j\frac{2\pi}{\lambda}\|\mathbf{t}_m-\mathbf{r}^U_{L_U}\|_2}\right]^{\rm T}.
\end{equation}
Then, by denoting $\mathbf{t} = [\mathbf{t}_1, \cdots, \mathbf{t}_M]$ as the position vector of the MA, the channel from the BS to the user is represented by
\begin{equation} \label{eq::legitimate_channel}
	\mathbf{H}(\mathbf{t}) = \left[\mathbf{h}(\mathbf{t}_1), \cdots, \mathbf{h}(\mathbf{t}_M)\right],
\end{equation}
where $\mathbf{h}(\mathbf{t}_m) \!=\! \Big[ g_{1,m}e^{j\frac{2\pi}{\lambda}\!\|\mathbf{t}_m-\mathbf{r}^U_1\|_2}, \!\cdots\!,$ $g_{L_U,m}e^{j\frac{2\pi}{\lambda}\!\|\mathbf{t}_m-\mathbf{r}^U_{L_U}\|_2}\Big]^{\rm H}$, and $g_{l,m}$ represents the path coefficient between the $m$-th MA at the BS and the $l$-th antenna of the user with $1 \leq l \leq L_U, 1 \leq m \leq M$. The near-field channel from the BS to the eavesdropper, denoted as $\mathbf{Z}(\mathbf{t})$, can be derived using a similar method, which is omitted for simplicity.
In this paper, to provide a secure MIMO communication performance upper bound for resilient designs, it is assumed that the BS has access to the perfect CSI of both the user and eavesdropper. This information can be obtained through effective beam training and/or near-field channel estimation techniques \cite{cui2022channel, lu2023near, cui2022near}.

\textit{Remarks:}  Unlike the far-field channel model, which captures only correlations in the angular domain \cite{zhao2024performance}, the near-field channel model presented in \eqref{eq::legitimate_channel} incorporates variations in both distance and angle for each antenna’s signal propagation paths to the user. As such, beam focusing can be applied to enhance secure transmission by exploiting channel variations between $\mathbf{H}(\mathbf{t})$ and $\mathbf{Z}(\mathbf{t})$.

\subsection{Problem Formulation}
Let $\mathbf{s} \in \mathbb{C}^{K}$ denote $K$ data streams transmitted from the BS to the user using $N$ transmit RF chains. As such, the received signal at the legitimate user is obtained as
\begin{equation} \label{eq::user_receive_signal}  \small
	\mathbf{y}_U = \mathbf{H}(\mathbf{t})\mathbf{W}_A\mathbf{W}_D\mathbf{s} + \mathbf{n}_U, 
\end{equation}  
where $\mathbf{W}_D \triangleq [\mathbf{w}_1, \cdots, \mathbf{w}_K] \in \mathbb{C}^{N \times K}$ and $\mathbf{W}_A \in \mathbb{C}^{M \times N}$ represent the digital and analog beamformers, respectively. Moreover, the element $w_{i,j}$ in the $i$-th row and $j$-th column of $\mathbf{W}_A$ must adhere to the unit-modulus condition, which is given by
\begin{equation} \label{eq::analog_beamformer} \small
	w_{i,j} \in \mathcal{A} \triangleq \left\{r^{j\theta}|\theta \in (0,2\pi]\right\}.
\end{equation}
Besides, $\mathbf{n}_U \sim \mathcal{CN}(\mathbf{0},\sigma^2_U\mathbf{I}_{L_U})$ represents the additive white Gaussian noise (AWGN) of the user. Thus, the achievable rate between the BS and the user is obtained as
\begin{equation} \small \label{eq::legitimate_SINR}
	R_U = \log_2\det\left(\mathbf{I}_{L_U} + \sigma_U^{-2}\mathbf{H}(\mathbf{t})\mathbf{W}_A\mathbf{W}_D\mathbf{W}_D^{\rm H}\mathbf{W}_A^{\rm H}\mathbf{H}^{\rm H}(\mathbf{t})\right).
\end{equation}
Similarly, the received signal at the eavesdropper is given by
\begin{equation} \label{eq::EVE_receive_signal}  \small
	\mathbf{y}_E = \mathbf{Z}(\mathbf{t})\mathbf{W}_A\mathbf{W}_D\mathbf{s} + \mathbf{n}_E, 
\end{equation} 
where $\mathbf{n}_E \sim \mathcal{CN}(\mathbf{0},\sigma^2_E\mathbf{I}_{L_E})$ denotes the AWGN at the eavesdropper. Thus, the achievable rate between the BS and eavesdropper is obtained as
\begin{equation} \small \label{eq::EVE_SINR}
	R_E = \log_2\det\left(\mathbf{I}_{L_E} + \sigma_E^{-2}\mathbf{Z}(\mathbf{t})\mathbf{W}_A\mathbf{W}_D\mathbf{W}_D^{\rm H}\mathbf{W}_A^{\rm H}\mathbf{Z}^{\rm H}(\mathbf{t})\right).
\end{equation}

As such, the secrecy rate is obtained as $R = [R_U - R_E]^+$ \cite{zhang2024physical}. In this paper, we aim to maximize the secrecy rate of the considered MA-aided near-field MIMO communication system by jointly optimizing the positions of MAs $\{\mathbf{t}_m\}_{m=1}^{M}$, the digital beamformer $\mathbf{W}_D$, and the analog beamformer $\mathbf{W}_A$. The corresponding problem is formulated as

\vspace{-1 em}
\begin{subequations} \label{eq::problem} \small
	\begin{align}
	 \max_{\mathbf{W}_D,\mathbf{W}_A,\{\mathbf{t}_m\}_{m=1}^{M}} &\quad  R \label{eq::problem_obj}\\
		{\rm s.t.}~ & \left\|\mathbf{W}_A\mathbf{W}_D\right\|_F^2 \leq P_B, \label{eq::problem_cons_pow}\\
		& \mathbf{t}_{m} \in \mathcal{C}, \: 1 \leq m \leq M, \label{eq::problem_cons_MA}\\
		& \|\mathbf{t}_m-\mathbf{t}_{\hat{m}}\|_2 \geq d_{\min}, 1 \leq m \neq 
		\hat{m} \leq M, \label{eq::problem_cons_MAdist}\\
		& w_{i,j} \in \mathcal{A}, 1 \leq i \leq M, 1 \leq j \leq N, \label{eq::problem_cons_unit-modulus}
	\end{align}
\end{subequations}
where constraint \eqref{eq::problem_cons_pow} represents the power budget at the BS, with $P_B$ being the maximum transmit power; constraint \eqref{eq::problem_cons_MA} specifies the antenna moving region; constraint \eqref{eq::problem_cons_MAdist} guarantees the minimum distance between MAs at the BS no less than the threshold, $d_{\min}$; and \eqref{eq::problem_cons_unit-modulus} is the unit-modulus constraints of the analog beamformer.
Problem \eqref{eq::problem} is non-convex with tightly coupled optimization variables, $\mathbf{W}_D$, $\mathbf{W}_A$, and $\{\mathbf{t}_m\}$, making it challenging to obtain globally optimal solutions within polynomial time.
Thus, we develop an AO-based solution to derive suboptimal solutions for problem \eqref{eq::problem}.

\section{Proposed Solution}
We first decompose the formulated problem into two independent subproblems, which are solved iteratively. Specifically, we develop a two-stage algorithm for the hybrid beamformer design subproblem. During the first stage, the WMMSE and block coordinate descent (BCD) algorithms are used to determine the fully-digital beamformers. During the second stage, the digital beamformer is obtained in closed form, while the analog beamformer is determined using the MO method. Then, for the MA position design subproblem, the MM algorithm is applied to optimize the position of each MA iteratively while maintaining the positions of others fixed.

\subsection{Hybrid Beamforming Design}
Given $\{\mathbf{t}_m\}$, the subproblem for designing hybrid beamforming at the BS can be reformulated as

\vspace{-1 em}
\begin{subequations} \label{eq::subProblem1} \small
	\begin{align}
		 \max_{\mathbf{W}_D,\mathbf{W}_A} \quad & R \label{eq::subProblem1_obj}\\
		{\rm s.t.}~ & \eqref{eq::problem_cons_pow}, \eqref{eq::problem_cons_unit-modulus}.
	\end{align}
\end{subequations}
To provide an upper bound on the secrecy rate performance for the considered hybrid beamforming system, we first design the fully-digital beamformer, denoted by $\mathbf{W} \in \mathbb{C}^{M \times K}$. 

\subsubsection{Stage I: Fully-Digital Beamforming Design} 
The unit-modulus constraints for the analog beamformer in \eqref{eq::problem_cons_unit-modulus} are neglected, and the fully-digital beamforming design problem is reformulated as

\vspace{-1 em}
\begin{subequations} \label{eq::subProblem11_FD} \small
	\begin{align}
		\max_{\mathbf{W}} \quad & \log_2\det\left(\mathbf{I}_{L_U} + \tilde{\mathbf{H}}(\mathbf{t})\mathbf{W}\mathbf{W}^{\rm H}\tilde{\mathbf{H}}^{\rm H}(\mathbf{t})\right)   \nonumber \\
		 &\!\!\!\!\!\!- \log_2\det\left(\mathbf{I}_{L_E} + \tilde{\mathbf{Z}}(\mathbf{t})\mathbf{W}\mathbf{W}^{\rm H}\tilde{\mathbf{Z}}^{\rm H}(\mathbf{t})\right) \label{eq::subProblem1_FD_obj} \\
		{\rm s.t.}~ & 
		\left\|\mathbf{W}\right\|_F^2 \leq P_B,
	\end{align}
\end{subequations}
where we define $\tilde{\mathbf{H}}(\mathbf{t}) = \mathbf{H}(\mathbf{t})\sigma^{-1}_U$ and $\tilde{\mathbf{Z}}(\mathbf{t}) = \mathbf{Z}(\mathbf{t})\sigma_E^{-1}$, respectively. 
Note that problem \eqref{eq::subProblem11_FD} is non-convex due to the intractability of the expression in \eqref{eq::subProblem1_FD_obj}. Thus, we develop a WMMSE-based method, in which the main idea is to transform the original rate maximization problem into another equivalent problem by introducing auxiliary variables. Specifically, this idea can be summarized in Theorem 1.

\begin{theorem}
By introducing auxiliary matrices $\mathbf{P} \in \mathbf{C}^{L_U \times K}$, $\mathbf{Q}_U \in \mathbb{C}^{K \times K}$, $\mathbf{Q}_E \in \mathbb{C}^{L_E \times L_E}$, and defining a matrix function $\mathbb{E}(\mathbf{P}, \mathbf{W}) \triangleq (\mathbf{I} - \mathbf{P}^{\rm H}\tilde{\mathbf{H}}(\mathbf{t})\mathbf{W})(\mathbf{I} - \mathbf{P}^{\rm H}\tilde{\mathbf{H}}(\mathbf{t})\mathbf{W})^{\rm H} + \mathbf{P}^{\rm H}\mathbf{P}$, the following optimization problem

\vspace{-1 em}
\begin{subequations} \label{eq::subProblem1_FD_transform} \small
	\begin{align}
		\max_{\mathbf{P}, \mathbf{Q}_U \succ 0, \mathbf{Q}_E \succ 0, \mathbf{W}}   &\log\det\left(\mathbf{Q}_U\right) - {\rm Tr}\left(\mathbf{Q}_U \mathbb{E}(\mathbf{P}, \mathbf{W})\right) + K \nonumber \\
		+\log\det\left(\mathbf{Q}_E\right) \!-\!& {\rm Tr}\!\left(\mathbf{Q}_E(\mathbf{I}_{L_E} \!+ \! \tilde{\mathbf{Z}}(\mathbf{t})\mathbf{W}\mathbf{W}^{\rm H})\tilde{\mathbf{Z}}^{\rm H}(\mathbf{t})\!\right) + L_E \label{eq::subProblem1_FD_transform_obj} \\
		 {\rm s.t.}~ &
		{\rm Tr}\left(\mathbf{W}\mathbf{W}^{\rm H}\right) \leq P_B,
	\end{align}
\end{subequations}
is equivalent to the secrecy rate maximization problem in \eqref{eq::problem}, and hence, the global optimal solutions for both problems are the same.
\end{theorem}

\begin{proof}
	Refer to Appendix \ref{app0}.
\end{proof}

Although the transformed problem \eqref{eq::subProblem1_FD_transform} introduces a greater number of variables compared to the original problem \eqref{eq::subProblem11_FD}, it is more tractable. This reformulation enables the use of the BCD method, which iteratively optimizes the objective function w.r.t. one variable block while keeping the others fixed. Specifically,  optimization variables are divided into three blocks, i.e., $\mathbf{P}$, $\{\mathbf{Q}_U, \mathbf{Q}_E\}$, and $\mathbf{W}$.
First, by fixing $\{\mathbf{Q}_U, \mathbf{Q}_E\}$ and $\mathbf{W}$, the problem of	 solving $\mathbf{P}$ is equivalent to minimizing ${\rm Tr}\left(\mathbf{Q}_U \mathbb{E}(\mathbf{P}, \mathbf{W})\right)$. According to \eqref{eq::optimalP}, the optimal solution of $\mathbf{P}$ is given by
\begin{equation} \small \label{eq::optimal_P}
	\mathbf{P}^{\star} = \left(\mathbf{I}_{L_U} + \tilde{\mathbf{H}}(\mathbf{t})\mathbf{W}\mathbf{W}^{\rm H}\tilde{\mathbf{H}}^{\rm H}(\mathbf{t})\right)^{-1}\tilde{\mathbf{H}}(\mathbf{t})\mathbf{W}.
\end{equation}
Second, by fixing $\mathbf{P}$ and $\mathbf{W}$, the problem of solving $\{\mathbf{Q}_U, \mathbf{Q}_E\}$ can be reduced to two separate subproblems, i.e., $\max_{\mathbf{Q}_U \succ 0}   \log\det\left(\mathbf{Q}_U\right) - {\rm Tr}\left(\mathbf{Q}_U \mathbb{E}(\mathbf{P}, \mathbf{W})\right)$ and $\max_{\mathbf{Q}_E \succ 0}\log\det\left(\mathbf{Q}_E\right) \!-\! {\rm Tr}\!(\mathbf{Q}_E(\mathbf{I}_{L_E} \!+ \! \tilde{\mathbf{Z}}(\mathbf{t})\mathbf{W}\mathbf{W}^{\rm H})\tilde{\mathbf{Z}}^{\rm H}(\mathbf{t})\!)$. As such, according to \eqref{eq::fact1}, the optimal solutions for $\mathbf{Q}_U$ and $\mathbf{Q}_E$ are given by
\begin{equation} \small \label{eq::optimal_QU}
	\mathbf{Q}_U^{\star} = \mathbb{E}(\mathbf{P}^{\star}, \mathbf{W})^{-1},
\end{equation}
and
\begin{equation} \small \label{eq::optimal_QE}
	\mathbf{Q}_E^{\star} = \left(\mathbf{I}_{L_E} \!+ \! \tilde{\mathbf{Z}}(\mathbf{t})\mathbf{W}\mathbf{W}^{\rm H}\tilde{\mathbf{Z}}^{\rm H}(\mathbf{t})\right)^{-1}.
\end{equation}

Third, by fixing $\mathbf{P}$ and $\{\mathbf{Q}_U, \mathbf{Q}_E\}$, the problem for solving $\mathbf{W}$ is equivalent to

\vspace{-1 em}
\begin{subequations}  \small \label{eq::subProblem1_FD_trans}
	\begin{align}
		\max_{\mathbf{W}} \quad &{\rm Tr}\left(\mathbf{Q}_U \mathbb{E}(\mathbf{P}, \mathbf{W})\right) + {\rm Tr}\!\left(\mathbf{Q}_E\left(\mathbf{I}_{L_E} \!+ \! \tilde{\mathbf{Z}}(\mathbf{t})\mathbf{W}\mathbf{W}^{\rm H}\right)\tilde{\mathbf{Z}}^{\rm H}(\mathbf{t})\!\right) \\
		 {\rm s.t.}~ &
		{\rm Tr}\left(\mathbf{W}\mathbf{W}^{\rm H}\right) \leq P_B,
	\end{align}
\end{subequations}
which is a second-order cone programming (SOCP) problem. Thus, problem \eqref{eq::subProblem1_FD_trans} is convex, and an optimal solution can be obtained. Nevertheless, the considered MA-aided near-field MIMO communication system typically involves a large number of antennas. This suggests that numerically solving problem \eqref{eq::subProblem1_FD_trans} may result in a high computational complexity. Given that problem \eqref{eq::subProblem1_FD_trans} is convex and adheres to Slater's condition, there is a strong duality between the primal and dual problems. Therefore, solving the dual problem provides the optimal solution to the primal problem \eqref{eq::subProblem1_FD_trans}. To address the constraint in problem \eqref{eq::subProblem1_FD_trans}, we introduce the Lagrange multiplier $\mu$, and the Lagrangian function w.r.t. $\mathbf{W}$ is then expressed as
\begin{equation} \small
	\begin{aligned}
		&\mathcal{L}(\mathbf{W},\mu) =  {\rm Tr}\left(\mathbf{Q}_E\left(\mathbf{I}_{L_E} + \tilde{\mathbf{Z}}(\mathbf{t})\mathbf{W}\mathbf{W}^{\rm H}\right)\tilde{\mathbf{Z}}^{\rm H}(\mathbf{t})\right) \\
		& + {\rm Tr}\left(\mathbf{Q}_U \mathbb{E}(\mathbf{P}, \mathbf{W})\right) + \mu \left({\rm Tr}\left(\mathbf{W}\mathbf{W}^{\rm H}\right) - P_B\right) \\
		& \quad\quad\quad\quad = {\rm Tr}\left(\mathbf{W}^{\rm H} \tilde{\mathbf{H}}^{\rm H}(\mathbf{t})\mathbf{P}\mathbf{Q}_U\mathbf{P}^{\rm H}\tilde{\mathbf{H}}(\mathbf{t})\mathbf{W}\right) \\
		& - {\rm Tr}\left(\mathbf{Q}_U\mathbf{P}^{\rm H}\tilde{\mathbf{H}}(\mathbf{t})\mathbf{W}\right)  - {\rm Tr}\left(\mathbf{Q}_U\mathbf{W}^{\rm H}\tilde{\mathbf{H}}^{\rm H}(\mathbf{t})\mathbf{P}\right) \\
		& + {\rm Tr}\left(\mathbf{W}^{\rm H}\tilde{\mathbf{Z}}^{\rm H}(\mathbf{t})\mathbf{Q}_E\tilde{\mathbf{Z}}(\mathbf{t})\mathbf{W}\right)  + \mu \left({\rm Tr}\left(\mathbf{W}\mathbf{W}^{\rm H}\right) - P_B\right).
	\end{aligned}
\end{equation}
Thus, the dual problem for \eqref{eq::subProblem1_FD_trans} can be expressed as

\vspace{-1 em}
\begin{subequations} \label{eq::dual_problem}
	\begin{align}
		\max_{\mu} \quad &f_1(\mu) \\
		 {\rm s.t.}~ &\mu \geq 0,
	\end{align}
\end{subequations}
where $f_1(\mu) \triangleq \min_{\mathbf{W}}\mathcal{L}(\mathbf{W},\mu)$. It is noted that \eqref{eq::dual_problem} is a linearly constrained convex quadratic optimization problem. The optimality condition for $\mathcal{L}(\mathbf{W},\mu)$ w.r.t. $\mathbf{W}$ is derived by equating its first-order derivative to zero, resulting in the closed-form solution expressed as
\begin{equation} \small \label{eq::optimalW}
	\begin{aligned}
		&\mathbf{W}^{\star}(\mu) = \Big(\tilde{\mathbf{H}}^{\rm H}(\mathbf{t})\mathbf{P}\mathbf{Q}_U\mathbf{P}^{\rm H}\tilde{\mathbf{H}}(\mathbf{t}) \\& \quad\quad\quad\quad\quad\quad\quad+ \mu\mathbf{I}_M + \tilde{\mathbf{Z}}^{\rm H}(\mathbf{t})\mathbf{Q}_E\tilde{\mathbf{Z}}(\mathbf{t})\Big)^{-1}\tilde{\mathbf{H}}^{\rm H}(\mathbf{t})\mathbf{P}\mathbf{Q}_U^{\rm H},
	\end{aligned}
\end{equation}
where the parameter $\mu$ is selected to satisfy the complementary slackness condition associated with the power constraint
\begin{equation}
	\mu \left({\rm Tr}\left(\mathbf{W}^{\star}(\mu)\mathbf{W}^{\star,\rm H}(\mu)\right) - P_B\right) = 0.
\end{equation}
Next, we examine whether $\mu = 0$ is indeed the optimal solution. If the condition ${\rm Tr}\left(\mathbf{W}^{\star}(0)\mathbf{W}^{\star,\rm H}(0)\right) - P_B \leq 0$ holds, the optimal solution is given by $\mathbf{W}^{\star}(0)$.
Otherwise, the optimal value of $\mu^{\star}$ should be positive and satisfy 
\begin{equation} \small \label{eq::positive_dualVar}
	{\rm Tr}\left(\mathbf{W}^{\star}(\mu)\mathbf{W}^{\star,\rm H}(\mu)\right) = P_B.
\end{equation}
Furthermore, it can be verified that the left-hand side of $\eqref{eq::positive_dualVar}$ decreases monotonically w.r.t. $\mu$ for $\mu \geq 0$. Consequently, $\mu^{\star}$ can be determined using a straightforward bisection search method. Finally, by substituting $\mu^{\star}$ into \eqref{eq::optimalW}, an optimal solution for $\mathbf{W}^{\star}(\mu^{\star})$ is obtained.

\subsubsection{Stage II: Hybrid Beamforming Design} 
Building upon $\mathbf{W}^{\star}$ derived in Stage I, we now proceed to address the design of the hybrid beamformer while taking hardware constraints into account. To approach the performance of achieved by using the fully-digital beamformer, we construct an optimization problem to design both the baseband digital beamformer $\mathbf{W}_D$ and analog beamformer $\mathbf{W}_A$, aiming to achieve near-optimal performance as follows

\vspace{-1 em}
\begin{subequations} \label{eq::subProblem2_HB}
	\begin{align}
		\min_{\mathbf{W}_A,\mathbf{W}_D} & \left\|\mathbf{W} - \mathbf{W}_A\mathbf{W}_D\right\|_F^2 \\
		 {\rm s.t.}~& \eqref{eq::problem_cons_unit-modulus},
	\end{align}	
\end{subequations}
which is a highly coupled quadratic problem. As such, an AO algorithm is developed to iteratively optimize $\mathbf{W}_A$ and $\mathbf{W}_D$.
On the one hand, with the fixed $\mathbf{W}_A$, the digital beamformer optimization problem can be reformulated as
\begin{equation} \label{eq::WD_solver}
	\min_{\mathbf{W}_D} \left\|\mathbf{W} - \mathbf{W}_A\mathbf{W}_D\right\|_F^2,
\end{equation}
which is a typical least-square (LS) problem. Thus, to derive the optimal digital beamformer $\mathbf{W}_D^{\star}$, we set the first-order derivative of \eqref{eq::WD_solver} w.r.t. $\mathbf{W}_D$ equal to zero, resulting in
\begin{equation} \label{eq::optimalWD}
	\mathbf{W}_D^{\star} = \left(\mathbf{W}_A^{\rm H}\mathbf{W}_A\right)^{-1}\mathbf{W}_A^{\rm H}\mathbf{W}.
\end{equation}
On the other hand, with the fixed $\mathbf{W}_D$, the analog beamformer optimization problem can be expressed as

\vspace{-1 em}
\begin{subequations} \label{eq::proble_analog}
	\begin{align}
		\min_{\mathbf{W}_A} & \left\|\mathbf{W} - \mathbf{W}_A\mathbf{W}_D\right\|_F^2 \\
		 {\rm s.t.}~& \eqref{eq::problem_cons_unit-modulus}.
	\end{align}	
\end{subequations}
Problem \eqref{eq::proble_analog} can be transformed into its equivalent form, which is given by

\vspace{-1 em}
\begin{subequations}  \small \label{eq::proble_analog_transform}
	\begin{align}
		\min_{\mathbf{W}_A} &  f_2\left({\rm vec}(\mathbf{W}_A)\right) \!\triangleq\! \left\| {\rm vec}(\mathbf{W}) \!-\! \left(\!\mathbf{W}_D^{\rm T} \otimes \mathbf{I}_N {\rm vec}(\mathbf{W}_A)\right) \right\|^2 \label{eq::proble_analog_transform_obj}\\
		 {\rm s.t.}~& \eqref{eq::problem_cons_unit-modulus}.
	\end{align}	
\end{subequations}
It is noted that problem \eqref{eq::proble_analog_transform} is non-convex due to the unit modulus constraints.
Thus, there is no general approach to optimally address problem \eqref{eq::proble_analog_transform}. 
Fortunately, \eqref{eq::proble_analog_transform_obj} is continuous and differentiable w.r.t. ${\rm vec}(\mathbf{W}_A)$. Moreover, vector $\mathbf{x} \triangleq {\rm vec}(\mathbf{W}_A)$ forms a Riemannian manifold, i.e., $\mathcal{N} = \{\mathbf{x} \in \mathbb{C}^{MN}: |\mathbf{x}_1|=|\mathbf{x}_2|=\cdots=|\mathbf{x}_{MN}|=1\}$. Therefore, we develop an efficient MO technique \cite{yu2016alternating, zargari2024riemannian} to derive a near-optimal solution. Specifically, in each iteration, the proposed method performs three main steps to find the optimal $\mathbf{W}_A$, as detailed below.

First, we need to determine the tangent space and Riemannian gradient. Specifically, the tangent space at a given point in the $t$-th iteration $\mathbf{x}^{\langle t \rangle} \in \mathcal{N}$ is given by
\begin{equation}
	\mathcal{T}_{\mathbf{x}^{\langle t \rangle}}\mathcal{N} = \left\{\mathbf{w} \in \mathbb{C}^{MN} | \Re\{\mathbf{w} \circ \mathbf{x}^{\langle t \rangle, *}\} = \mathbf{0}_{MN}\right\},
\end{equation}
where $\mathbf{w}$ denotes the tangent vector at $\mathbf{x}^{\langle t \rangle}$.
Next, the Riemannian gradient ${\rm grad} f_2(\mathbf{x}^{\langle t \rangle})$, which represents the orthogonal projection of the Euclidean gradient $\nabla f_2(\mathbf{x}^{\langle t \rangle})$ of \eqref{eq::proble_analog_transform_obj} onto the tangent space, is represented by
\begin{equation} \small \label{eq::Riemannian_gradient}
		{\rm grad} f_2(\mathbf{x}^{\langle t \rangle}) = \nabla f_2(\mathbf{x}^{\langle t \rangle}) - \Re\left\{\nabla f_2(\mathbf{x}^{\langle t \rangle}) \circ \mathbf{x}^{{\langle t\rangle},*}\right\} \circ \mathbf{x}^{\langle t \rangle},
\end{equation}
where $\nabla f_2(\mathbf{x}^{\langle t \rangle}) = 2(\mathbf{W}_D^{\rm T} \otimes \mathbf{I}_N)^{\rm H} (\mathbf{W}_D^{\rm T} \otimes \mathbf{I}_N)\mathbf{x}^{\langle t \rangle} - 2(\mathbf{W}_D^{\rm T} \otimes \mathbf{I}_N)^{\rm H}{\rm vec}(\mathbf{W})$.

Second, we determine the transport vector and search direction on $\mathcal{N}$. Specifically, the search direction is calculated by the tangent vector conjugated to ${\rm grad} f_2(\mathbf{x})$, which can be represented by
\begin{equation} \label{eq::direction}
	\mathbf{d}^{\langle t+1 \rangle} = -{\rm grad} f_2(\mathbf{x}^{\langle t+1 \rangle}) + \upsilon^{\langle t+1 \rangle} \mathcal{T}_{\mathbf{x}^{\langle t \rangle} \mapsto \mathbf{x}^{\langle t+1 \rangle}}(\mathbf{d}^{\langle t \rangle}),
\end{equation}
where $\mathbf{d}^{\langle t \rangle}$ is the search direction for $\mathbf{x}^{\langle t+1 \rangle}$. Since $\mathbf{d}^{\langle t \rangle}$ and $\mathbf{d}^{\langle t+1 \rangle}$
belong to different spaces, i.e., $\mathcal{T}_{\mathbf{x}^{\langle t \rangle}}\mathcal{N}$ and $\mathcal{T}_{\mathbf{x}^{\langle t+1 \rangle}}\mathcal{N}$, the search direction on $\mathcal{N}$ cannot be directly derived. Thus, a transport operation is applied to resolve this issue, which maps $\mathbf{d}^{\langle t \rangle}$ onto $\mathcal{T}_{\mathbf{x}^{\langle t+1 \rangle}}\mathcal{N}$. Specifically, the transport of a tangent vector $\mathbf{d}^{\langle t \rangle}$ from $\mathcal{T}_{\mathbf{x}^{\langle t \rangle}}\mathcal{N}$ to $\mathcal{T}_{\mathbf{x}^{\langle t+1 \rangle}}\mathcal{N}$ is represented by
\begin{equation} \small \label{eq::transport_operation}
	\mathcal{T}_{\mathbf{x}^{\langle t \rangle} \mapsto \mathbf{x}^{\langle t+1 \rangle}}(\mathbf{d}^{\langle t \rangle}) = \mathbf{d}^{\langle t \rangle} - \Re\{\mathbf{d}^{\langle t \rangle} \circ \mathbf{x}^{{\langle t+1 \rangle}, *} \} \circ \mathbf{x}^{\langle t+1 \rangle}. 
\end{equation}
Besides, $\upsilon^{\langle t+1 \rangle}$ represents the Polak-Ribiere parameter, which is utilized to accelerate the convergence rate \cite{pan2020intelligent}, and it can be given by
\begin{equation} \small \label{eq::Polak-Ribiere}
	\upsilon^{\langle t+1 \rangle} \!=\! \frac{\Re\{{\rm grad}^{\rm H} f_2(\mathbf{x}^{\langle t+1 \rangle})[{\rm grad} f_2(\mathbf{x}^{\langle t+1 \rangle}) - {\rm grad} f_2(\mathbf{x}^{\langle t \rangle})]\}}{{\rm grad}^{\rm H} f_2(\mathbf{x}^{\langle t \rangle}) \: {\rm grad} f_2(\mathbf{x}^{\langle t \rangle})}.
\end{equation}

Third, to ensure that the updated point remains on the manifold, a retraction step should be applied. The retraction process maps the current point $\mathbf{x}^{\langle t \rangle}$ from the tangent space onto the complex circle manifold, which is represented by
\begin{equation} \small \label{eq::retraction}
	\begin{aligned}
		&{\rm Retr}_{\mathbf{x}^{\langle t \rangle}}(\beta^{\langle t \rangle} \mathbf{d}^{\langle t \rangle}) \\
		&\triangleq \mathcal{T}_{\mathbf{x}^{\langle t \rangle}}\mathcal{N} \mapsto \mathcal{N}: \left[\beta^{\langle t \rangle} \mathbf{d}^{\langle t \rangle}\right]_i \mapsto \frac{\left[\mathbf{x}^{\langle t \rangle} + \beta^{\langle t \rangle} \mathbf{d}^{\langle t \rangle}\right]_i}{\left[\left|\mathbf{x}^{\langle t \rangle} + \beta^{\langle t \rangle} \mathbf{d}^{\langle t \rangle}\right|\right]_i},
	\end{aligned}	
\end{equation}
where $\beta^{\langle t \rangle}$ denotes the Armijo backtracking step size.

The steps of the proposed MO-based algorithm for obtaining the optimal analog beamformer, $\mathbf{W}_A^{\star}$, are outlined in Algorithm~\ref{alg_MO}, with $\epsilon_1$ denoting the convergence accuracy. Besides, according to \cite{alhujaili2019transmit}, Algorithm~\ref{alg_MO} is guaranteed to converge. 

\begin{algorithm}[t] \small
	\caption{MO Algorithm for Problem \eqref{eq::proble_analog_transform}}
	\label{alg_MO}
	\begin{algorithmic}[1]
		\STATE \textbf{Initialize} a feasible initial point $\mathbf{x}^{\langle 0 \rangle}$, compute $\mathbf{d}^{\langle 0 \rangle} = - {\rm grad} f_2(\mathbf{x}^{\langle 0 \rangle})$,  and set the iteration index $t=0$.
		\REPEAT
		\STATE Choose the Armijo line search step $\beta^{\langle t \rangle}$, and update $\mathbf{x}^{\langle t+1 \rangle}$ via the retraction mapping as defined in \eqref{eq::retraction}.
		\STATE Calculate Riemannian gradient ${\rm grad} f_2(\mathbf{x}^{\langle t+1 \rangle})$ based on \eqref{eq::Riemannian_gradient}.
		\STATE Determine transport $\mathcal{T}_{\mathbf{x}^{\langle t \rangle} \mapsto \mathbf{x}^{\langle t+1 \rangle}}(\mathbf{d}^{\langle t \rangle})$ according to \eqref{eq::transport_operation}.
		\STATE Update the Polak-Ribiere parameter $\upsilon^{\langle t+1 \rangle}$ based on \eqref{eq::Polak-Ribiere}.
		\STATE Update the conjugate direction $\mathbf{d}^{\langle t+1 \rangle}$ according to \eqref{eq::direction}.
		\STATE Set $t$ $\leftarrow$ $t+1$.
		\UNTIL $\|{\rm grad} f_2(\mathbf{x}^{\langle t \rangle}) \|_2 \leq \epsilon_1$ or $t \geq t_{\rm max}^{\rm MO}$.
		\RETURN $\mathbf{x}^{\star}$.
	\end{algorithmic}
\end{algorithm}

\subsection{MA Position Design}
Given $\mathbf{W}_D$ and $\mathbf{W}_A$, the subproblem for optimizing the MA positions is given by

\vspace{-1 em}
\begin{subequations}\small  \label{eq::subproblem2}
	\begin{align}
		\max_{\{\mathbf{t}_m\}} \quad  f_3(\mathbf{t}) &\triangleq \log_2\det\left(\mathbf{I}_{L_U} + \tilde{\mathbf{H}}(\mathbf{t})\mathbf{V}\mathbf{V}^{\rm H}\tilde{\mathbf{H}}^{\rm H}(\mathbf{t})\right)   \nonumber \\
		&- \log_2\det\left(\mathbf{I}_{L_E} + \tilde{\mathbf{Z}}(\mathbf{t})\mathbf{V}\mathbf{V}^{\rm H}\tilde{\mathbf{Z}}^{\rm H}(\mathbf{t})\right) \label{eq::subproblem2_obj}\\
		{\rm s.t.}~ & 
		\eqref{eq::problem_cons_MA}, \eqref{eq::problem_cons_MAdist},
	\end{align}
\end{subequations}
where we define $\mathbf{V} \triangleq \mathbf{W}_A\mathbf{W}_D$.
Based on Theorem 1, \eqref{eq::subproblem2_obj} can be simplified as
\begin{equation} \small \label{eq::subproblem2_transformed}
	\begin{aligned}
		&f_3(\mathbf{t}) = -2\Re\left\{{\rm Tr}\left(\mathbf{Q}_U\mathbf{V}^{\rm H}\tilde{\mathbf{H}}^{\rm H}(\mathbf{t})\mathbf{P}\right)\right\} \\
		& \!-\! {\rm Tr}\left(\mathbf{V}^{\rm H}\tilde{\mathbf{H}}^{\rm H}(\mathbf{t})\mathbf{P}\mathbf{Q}_U\mathbf{P}^{\rm H}\tilde{\mathbf{H}}(\mathbf{t})\mathbf{V}\right) 
		\!-\! {\rm Tr}\left(\mathbf{V}^{\rm H}\tilde{\mathbf{Z}}^{\rm H}(\mathbf{t})\mathbf{Q}_E\tilde{\mathbf{Z}}(\mathbf{t})\mathbf{V}\right),
	\end{aligned}
\end{equation}
where $\mathbf{P}$, $\mathbf{Q}_U$, and $\mathbf{Q}_E$ are obtained by replacing $\mathbf{W}$ with $\mathbf{V}$ in \eqref{eq::optimal_P}, \eqref{eq::optimal_QU}, and \eqref{eq::optimal_QE}, respectively.
Then, define $\mathbf{B} = \mathbf{P}\mathbf{Q}_U\mathbf{V}^{\rm H}$, $\mathbf{W}_X = \mathbf{V}\mathbf{V}^{\rm H}$, and $\mathbf{C} = \mathbf{P}\mathbf{Q}_U\mathbf{P}^{\rm H}$. Thus, \eqref{eq::subproblem2_transformed} can be further simplified as
\begin{equation} \small \label{eq::subproblem2_transformed2}
	\begin{aligned}
		f_3(\mathbf{t}) = &-2\Re\left\{{\rm Tr}\left(\mathbf{B}\tilde{\mathbf{H}}^{\rm H}(\mathbf{t})\right)\right\}  - {\rm Tr}\left(\tilde{\mathbf{H}}(\mathbf{t})\mathbf{W}_X\tilde{\mathbf{H}}^{\rm H}(\mathbf{t})\mathbf{C}\right) \\
		&- {\rm Tr}\left(\tilde{\mathbf{Z}}(\mathbf{t})\mathbf{W}_X\tilde{\mathbf{Z}}^{\rm H}(\mathbf{t})\mathbf{Q}_E\right).
	\end{aligned}
\end{equation}
To obtain a tractable solution, we adopt the BCD method to optimize the position of the $m$-th MA, $\mathbf{t}_m$, while keeping the positions of all other MAs, i.e., $\{\mathbf{t}_{\hat{m}} , m \neq \hat{m}\}$,  fixed.
Specifically, denote $\mathbf{b}_i \in \mathbb{C}^{L_U \times 1}$ as the $m$-th column of $i$-th column of $\mathbf{B}$, $1 \leq i \leq M$. Thus, the first term of \eqref{eq::subproblem2_transformed2} can be rewritten as
\begin{equation} \small \label{eq::first_term}
		2\Re\left\{{\rm Tr}\left(\mathbf{B}\tilde{\mathbf{H}}^{\rm H}(\mathbf{t})\right)\right\} \!=\! 2\Re\left\{\!\mathbf{h}^{\rm H}(\mathbf{t}_m)\mathbf{b}_m \!+\! \overbrace{\sum_{i \neq m}\mathbf{h}^{\rm H}(\mathbf{t}_i)\mathbf{b}_i}^{\rm const}\right\}.
\end{equation}
In addition, denote $\mathbf{v}_j^{\rm H} \in \mathbb{C}^{1 \times K}$, $1 \leq j \leq N$, as the $j$-th row of $\mathbf{V}$. Thus, the second term of \eqref{eq::subproblem2_transformed2} can be transformed into \eqref{eq::second_term} as shown at the top of the next page. Similarly, the third term of \eqref{eq::subproblem2_transformed2} is derived in \eqref{eq::third_term} as shown at the top of the next page. By substituting \eqref{eq::first_term}, \eqref{eq::second_term}, and \eqref{eq::third_term} into \eqref{eq::subproblem2_transformed2} and ignoring the constant terms w.r.t. $\mathbf{t}_m$, problem \eqref{eq::subproblem2} is simplified as
\begin{figure*} [ht] 
	\centering
	\vspace*{8pt} 
\begin{equation} \label{eq::second_term} \small
	\begin{aligned}
	{\rm Tr}&\left(\tilde{\mathbf{H}}(\mathbf{t})\mathbf{W}_X\tilde{\mathbf{H}}^{\rm H}(\mathbf{t})\mathbf{C}\right) = {\rm Tr}\left(\tilde{\mathbf{H}}(\mathbf{t})\mathbf{V}\mathbf{V}^{\rm H}\tilde{\mathbf{H}}^{\rm H}(\mathbf{t})\mathbf{C}\right) = {\rm Tr}\left[\left(\sum_{i=1}^{M} \mathbf{h}(\mathbf{t}_i) \mathbf{v}_i^{\rm H}\right)\left(\sum_{j=1}^{M} \mathbf{v}_j\mathbf{h}(\mathbf{t}_j)^{\rm H} \right)\mathbf{C}\right] \\
	&={\rm Tr}\left[\mathbf{h}(\mathbf{t}_m) \mathbf{v}_m^{\rm H} \mathbf{v}_m\mathbf{h}(\mathbf{t}_m)^{\rm H}\mathbf{C}\right] + {\rm Tr}\left[\mathbf{h}(\mathbf{t}_m) \mathbf{v}_m^{\rm H}\left(\sum_{i=1, i \neq m}^{M} \mathbf{v}_i\mathbf{h}(\mathbf{t}_i)^{\rm H} \right) \mathbf{C}\right] \\
	&+ {\rm Tr}\left[\left(\sum_{j=1, j \neq m}^{M} \mathbf{h}(\mathbf{t}_j) \mathbf{v}_j^{\rm H}\right)\mathbf{v}_m\mathbf{h}(\mathbf{t}_m)^{\rm H} \mathbf{C}\right] \underbrace{+{\rm Tr}\left[\left(\sum_{i=1, i \neq m}^{M} \mathbf{h}(\mathbf{t}_i) \mathbf{v}_i^{\rm H}\right)\left(\sum_{j=1, j \neq m}^{M} \mathbf{v}_j\mathbf{h}(\mathbf{t}_j)^{\rm H} \right) \mathbf{C}\right]}_{\rm const}.
	\end{aligned}
\end{equation}
%	\hrulefill
	\vspace{-3 em}
\end{figure*}
\begin{figure*} [ht] 
	\centering
	\vspace*{8pt} 
	\begin{equation} \label{eq::third_term} \small
		\begin{aligned}
		&{\rm Tr}\left(\tilde{\mathbf{Z}}(\mathbf{t})\mathbf{W}_X\tilde{\mathbf{Z}}^{\rm H}(\mathbf{t})\mathbf{Q}_E\right) = {\rm Tr}\left(\tilde{\mathbf{Z}}(\mathbf{t})\mathbf{V}\mathbf{V}^{\rm H}\tilde{\mathbf{Z}}^{\rm H}(\mathbf{t})\mathbf{Q}_E\right) = {\rm Tr}\left[\left(\sum_{i=1}^{M} \mathbf{z}(\mathbf{t}_i) \mathbf{v}_i^{\rm H}\right)\left(\sum_{j=1}^{M} \mathbf{v}_j\mathbf{z}(\mathbf{t}_j)^{\rm H} \right)\mathbf{Q}_E\right] \\
		&={\rm Tr}\left[\mathbf{z}(\mathbf{t}_m) \mathbf{v}_m^{\rm H} \mathbf{v}_m\mathbf{z}(\mathbf{t}_m)^{\rm H}\mathbf{Q}_E\right] + {\rm Tr}\left[\mathbf{z}(\mathbf{t}_m) \mathbf{v}_m^{\rm H}\left(\sum_{i=1, i \neq m}^{M} \mathbf{v}_i\mathbf{z}(\mathbf{t}_i)^{\rm H} \right) \mathbf{Q}_E\right] \\
		&+ {\rm Tr}\left[\left(\sum_{j=1, j \neq m}^{M} \mathbf{z}(\mathbf{t}_j) \mathbf{v}_j^{\rm H}\right)\mathbf{v}_m\mathbf{z}(\mathbf{t}_m)^{\rm H} \mathbf{Q}_E\right] \underbrace{+{\rm Tr}\left[\left(\sum_{i=1, i \neq m}^{M} \mathbf{z}(\mathbf{t}_i) \mathbf{v}_i^{\rm H}\right)\left(\sum_{j=1, j \neq m}^{M} \mathbf{v}_j\mathbf{z}(\mathbf{t}_j)^{\rm H} \right) \mathbf{Q}_E\right]}_{\rm const}.
		\end{aligned}
	\end{equation}
	\hrulefill
\end{figure*}

\vspace{-1 em}
\begin{subequations} \label{eq::subproblem2_transformed3} \small
	\begin{align}
		\min_{\mathbf{t}_m} \quad &f_4(\mathbf{t}_m) = \mathbf{h}^{\rm H}(\mathbf{t}_m)\mathbf{D}^U_m\mathbf{h}(\mathbf{t}_m) 
		+ \mathbf{z}^{\rm H}(\mathbf{t}_m)\mathbf{D}^E_m\mathbf{z}(\mathbf{t}_m)\nonumber\\  
		&\quad\quad\quad\quad+2\Re\left\{\mathbf{h}^{\rm H}(\mathbf{t}_m)\mathbf{\mathbf{r}}_m^U\right\} +2\Re\left\{\mathbf{z}^{\rm H}(\mathbf{t}_m)\mathbf{\mathbf{r}}_m^E\right\} \label{eq::subproblem2_transformed3_obj} \\
		{\rm s.t.}~ & 
		\eqref{eq::problem_cons_MA}, \eqref{eq::problem_cons_MAdist},
	\end{align}
\end{subequations}
where $\mathbf{D}^U_m \triangleq \mathbf{v}_m^{\rm H} \mathbf{v}_m\mathbf{C}$, $\mathbf{D}^E_m \triangleq \mathbf{v}_m^{\rm H} \mathbf{v}_m\mathbf{Q}_E$, and 
\begin{equation}
	\quad\quad\quad\: \mathbf{r}_m^U = \mathbf{C}\left(\sum_{i=1, i \neq m}^{M} \mathbf{h}(\mathbf{t}_i) \mathbf{v}_i^{\rm H}\right)\mathbf{v}_m + \mathbf{b}_m,
\end{equation}
\begin{equation}
	\mathbf{r}_m^E = \mathbf{Q}_E\left(\sum_{i=1, i \neq m}^{M} \mathbf{z}(\mathbf{t}_i) \mathbf{v}_i^{\rm H}\right)\mathbf{v}_m.
\end{equation}
Note that  \eqref{eq::subproblem2_transformed3_obj} remains non-convex. Therefore, the MM algorithm is employed to obtain a suboptimal solution. Specifically, to construct the surrogate function, the subsequent lemma is presented.

\begin{lemma}
	For any given antenna position in the $t$-th iteration, $\mathbf{t}_m^{\langle t \rangle}$, we have
	\begin{equation}
		\begin{aligned}
			&\mathbf{h}^{\rm H}(\mathbf{t}_m)\mathbf{D}^I_m\mathbf{h}(\mathbf{t}_m) \leq  \mathbf{h}^{\rm H}(\mathbf{t}_m)\mathbf{\Phi}^I_m\mathbf{h}(\mathbf{t}_m) \\
			& -2\Re\left\{\mathbf{h}^{\rm H}(\mathbf{t}_m)\left(\mathbf{\Phi}^I_m - \mathbf{D}^I_m\right)\mathbf{h}(\mathbf{t}_m^{\langle t \rangle})\right\} \\
			& + \mathbf{h}^{\rm H}(\mathbf{t}_m^{\langle t \rangle})\left(\mathbf{\Phi}^I_m - \mathbf{D}^I_m\right)\mathbf{h}(\mathbf{t}_m^{\langle t \rangle}) \triangleq \Xi^I\left(\mathbf{t}_m|\mathbf{t}_m^{\langle t \rangle}\right),
		\end{aligned} 	
	\end{equation}
	where $\mathbf{\Phi}^I_m = \zeta^I_{\max}\mathbf{I}_{L_U}$, and $\zeta^I_{\max}$ denotes the maximum eigenvalue of $\mathbf{D}^I_m$, $I \in \{U, E\}$.
\end{lemma}

\begin{proof}
	Please refer to \cite{song2015optimization}.
\end{proof}

Based on Lemma 1, the objective in \eqref{eq::subproblem2_transformed3} can be simplified by eliminating the constant terms, yielding
\begin{equation}	
	\begin{aligned}
		f_4(\mathbf{t})	& \leq \Xi^U\left(\mathbf{t}_m|\mathbf{t}_m^{\langle t \rangle}\right) + \Xi^E\left(\mathbf{t}_m|\mathbf{t}_m^{\langle t \rangle}\right) \\
			&  + 2\Re\left\{\mathbf{h}^{\rm H}(\mathbf{t}_m)\mathbf{\mathbf{r}}_m^U\right\} + 2\Re\left\{\mathbf{z}^{\rm H}(\mathbf{t}_m)\mathbf{r}_m^E\right\} \triangleq f_5(\mathbf{t}).
	\end{aligned}
\end{equation}
Next, since we have $\mathbf{h}^{\rm H}(\mathbf{t}_m)\mathbf{h}(\mathbf{t}_m) = \sum_{l=1}^{L_U}g^2_{l,m}$, and $\mathbf{h}^{\rm H}(\mathbf{t}_m)\mathbf{\Phi}^I_m\mathbf{h}(\mathbf{t}_m) = \sum_{l=1}^{L_U}g^2_{l,m}\lambda^I_{\max}$, which is a constant. Thus, $f_5(\mathbf{t})$ can be simplified as
\begin{equation} \label{eq::subproblem2_transformed4}
	\begin{aligned} 
		& f_5(\mathbf{t}_{m})= 2\Re\left\{\mathbf{h}^{\rm H}(\mathbf{t}_{m})\boldsymbol{\tau}_{m}^U+\mathbf{z}^{\rm H}(\mathbf{t}_{m})\boldsymbol{\tau}_{m}^E\right\} + {\rm const},
	\end{aligned}
\end{equation}
where
\begin{equation} \label{eq::tauU}
	\boldsymbol{\tau}_{m}^U = \mathbf{\mathbf{r}}_{m}^U - \left(\mathbf{\Phi}^U_{m} - \mathbf{D}^U_{m}\right)\mathbf{h}(\mathbf{t}_{m}^{\langle t \rangle}),
\end{equation}
and 
\begin{equation} \label{eq::tauE}
	\boldsymbol{\tau}_{m}^E = \mathbf{\mathbf{r}}_{m}^E - \left(\mathbf{\Phi}^E_{m} - \mathbf{D}^E_{m}\right)\mathbf{h}(\mathbf{t}_{m}^{\langle t \rangle}).
\end{equation}
Although the transformed \eqref{eq::subproblem2_transformed4} is more tractable than its original form in \eqref{eq::subproblem2_transformed}, $f_5(\mathbf{t}_{m})$ remains non-convex w.r.t. $\mathbf{t}_m$. As such, the second-order Taylor expansion of $f_5(\mathbf{t}_{m})$ is utilized to construct an upper bound surrogate function.
Specifically, by introducing $\delta_m >0$ and ensuring that $\delta_m\mathbf{I}_2 \succeq \nabla^2f_5(\mathbf{t}_{m})$, we have
\begin{equation} \small \label{eq::subproblem2_transformed5}
	\begin{aligned}
		f_5(\mathbf{t}_{m}) &\leq f_5(\mathbf{t}_{m}^{\langle t \rangle}) + \nabla f_5(\mathbf{t}_{m}^{\langle t \rangle})^{\rm T}(\mathbf{t}_{m} - \mathbf{t}_{m}^{\langle t \rangle}) \\
		& + \frac{\delta_{m}}{2}(\mathbf{t}_{m} - \mathbf{t}_{m}^{\langle t \rangle})^{\rm T}(\mathbf{t}_{m} - \mathbf{t}_{m}^{\langle t \rangle}) \\
		& = \frac{\delta_{m}}{2}\mathbf{t}_{m}^{\rm T}\mathbf{t}_{m} + \left(\nabla f(\mathbf{t}_{m}^{\langle t \rangle}) - \delta_{m}\mathbf{t}_{m}^{\langle t \rangle}\right)^{\rm T}\mathbf{t}_{m} +{\rm const},
	\end{aligned}
\end{equation}
where $\nabla f_5(\mathbf{t}_{m}^{\langle t \rangle})$ and $\nabla^2 f_5(\mathbf{t}_{m}^{\langle t \rangle})$ are derived in Appendix B. Besides, the construction of $\delta_{m}$ can be constructed as in \cite{maY2024movable}.

Although \eqref{eq::subproblem2_transformed5} is convex quadratic w.r.t. $\mathbf{t}_{m}$, the constraint \eqref{eq::problem_cons_MAdist} is non-convex. Thus, we relax \eqref{eq::problem_cons_MAdist} by deriving its first-order Taylor expansion as
\begin{equation}  \label{eq::problem_cons_MAdist_transformed}
	\begin{aligned}
		\left\|\mathbf{t}_m - \mathbf{t}_{\hat{m}}\right\|_2 &\geq \left\|\mathbf{t}_m^{\langle t \rangle} - \mathbf{t}_{\hat{m}}\right\|_2 \\&+ \left(\!\nabla\left\|\mathbf{t}_m - \mathbf{t}_{\hat{m}}\right\|_2|_{\mathbf{t}_m = \mathbf{t}_m^{\langle t \rangle}}\right)^{\rm T}(\mathbf{t}_m - \mathbf{t}_m^{\langle t \rangle})\\
		&= \frac{(\mathbf{t}_m^{\langle t \rangle} - \mathbf{t}_{\hat{m}})^{\rm T}}{\left\|\mathbf{t}_m^{\langle t \rangle} - \mathbf{t}_{\hat{m}}\right\|_2}(\mathbf{t}_m - \mathbf{t}_{\hat{m}}).
	\end{aligned}
\end{equation}

Finally, by substituting \eqref{eq::subproblem2_transformed5} and \eqref{eq::problem_cons_MAdist_transformed} into \eqref{eq::subproblem2_transformed3}, problem \eqref{eq::subproblem2_transformed3} can be reformulated as

\vspace{-1 em}
\begin{subequations} \label{eq::subproblem2_transformed_final}
	\begin{align}
		\min_{\mathbf{t}_m} \quad &\frac{\delta_{m}}{2}\mathbf{t}_{m}^{\rm T}\mathbf{t}_{m} + \left(\nabla f(\mathbf{t}_{m}^{\langle t \rangle}) - \delta_{m}\mathbf{t}_{m}^{\langle t \rangle}\right)^{\rm T}\mathbf{t}_{m} \\
		{\rm s.t.}\quad & \mathbf{t}_m \in \mathcal{C},\\
		& \frac{(\mathbf{t}_m^{\langle t \rangle} - \mathbf{t}_{\hat{m}})^{\rm T}}{\left\|\mathbf{t}_m^{\langle t \rangle} - \mathbf{t}_{\hat{m}}\right\|_2}(\mathbf{t}_m - \mathbf{t}_{\hat{m}}) \geq d_{\min},
	\end{align}
\end{subequations}
which is convex and can be optimally addressed by using interior point methods or standard convex optimization tools, e.g., CVX. Building on the preceding discussions, the comprehensive steps of the MM algorithm are outlined in Algorithm~\ref{alg::MM}.

\begin{algorithm}[t] \small
	\caption{MM Algorithm for Problem \eqref{eq::subproblem2_transformed3}}
	\label{alg::MM}
	\begin{algorithmic}[1]
		\STATE \textbf{Initialize} a feasible antenna position $\mathbf{t}_m^{\langle 0 \rangle}$, the convergence threshold $\epsilon_2$, the iteration index $t = 0$, and the maximum number of iterations $t_{\max}$.
		\REPEAT
		\STATE Calculate $\mathbf{\Phi}^I_m, I \in \{U, E\}$ based on Lemma 1.
		\STATE Calculate $\boldsymbol{\tau}_{m}^U, I \in \{U, E\}$ according to \eqref{eq::tauU} and \eqref{eq::tauE}.
		\STATE Calculate $\nabla f_5(\mathbf{t}_{m}^{\langle t \rangle})$ and $\nabla^2 f_5(\mathbf{t}_{m}^{\langle t \rangle})$  \eqref{eq::first-order} and \eqref{eq::second-order}.
		\STATE Obtain $\mathbf{t}_m^{(t+1)}$ by solving problem  \eqref{eq::subproblem2_transformed_final}.
		\STATE Set $t$ $\leftarrow$ $t+1$.
		\UNTIL $|f_4(\mathbf{t}_m^{\langle t \rangle}) - f_4(\mathbf{t}_m^{\langle t-1 \rangle})| / |f_4(\mathbf{t}_m^{\langle t-1 \rangle})| \leq \epsilon_2$ or $t \geq t_{\rm max}^{\rm MM}$.
		\RETURN $\mathbf{t}_m^{\star}$.
	\end{algorithmic}
\end{algorithm}

\subsection{Overall Solution}\label{SecC}

\begin{algorithm}[t] \small
	\caption{Overall Algorithm for Problem  \eqref{eq::problem}}
	\label{alg::overall}
	\begin{algorithmic}[1]
		\STATE  \textbf{Initialize}  $\mathbf{W}^{\langle 0 \rangle}$, $\mathbf{W}_D^{\langle 0 \rangle}$, $\mathbf{W}_A^{\langle 0 \rangle}$, $\mathbf{t}^{\langle 0 \rangle}$ the convergence accuracy $\epsilon_3$,  iteration index $T = 0$, and the maximum number  of iterations $T_{\rm max}$, calculate the value of \eqref{eq::problem_obj} as $F(\mathbf{W}_D^{\langle 0 \rangle}, \mathbf{W}_A^{\langle 0 \rangle},\mathbf{t}^{\langle 0 \rangle})$.
		\REPEAT
			\STATE Update the fully-digital beamformer $\mathbf{W}^{\langle T \rangle}$ as in \eqref{eq::optimalW}.
			\STATE Approximate the digital beamformer $\mathbf{W}_D^{\langle T \rangle}$ as in \eqref{eq::optimalWD}, and approximate the analog beamformer $\mathbf{W}_A^{\langle T \rangle}$ by Algorithm~\ref{alg_MO}.
			\FOR{$m=1$ to $M$}  
			\STATE Update  $\mathbf{t}_m^{\langle T \rangle}$ by  Algorithm~\ref{alg::MM}.
			\ENDFOR
			\STATE Set $T \leftarrow T + 1$.
		\UNTIL $\frac{\|F(\mathbf{W}_D^{\langle T \rangle}, \mathbf{W}_A^{\langle T \rangle},\mathbf{t}^{\langle T \rangle}) - F(\mathbf{W}_D^{\langle T-1 \rangle}, \mathbf{W}_A^{\langle T-1 \rangle},\mathbf{t}^{\langle T-1 \rangle})\|}{F(\mathbf{W}_D^{\langle T-1 \rangle}, \mathbf{W}_A^{\langle T-1 \rangle},\mathbf{t}^{\langle T-1 \rangle})} \leq \epsilon_3$ or $T > T_{\rm max}$.
		\RETURN $\mathbf{W}_D^{\star}$, $\mathbf{W}_A^{\star}$, and $\mathbf{t}^{\star}$.
	\end{algorithmic}
\end{algorithm}

We summarize the detailed procedures of the proposed overall solution in Algorithm~\ref{alg::overall}.

\subsubsection{Convergence Analysis}
Each subproblem is addressed locally and/or optimally in Algorithm~\ref{alg::overall}, which guarantees that \eqref{eq::problem_obj} does not decrease over iterations. Besides, due to the limitation imposed by finite transmission power in \eqref{eq::problem_obj}, Algorithm~\ref{alg::overall} ensures to be converged.

\subsubsection{Computational Complexity}
Then, we analyze the complexity of Algorithm~\ref{alg::overall} as follows.
First, in step 3, the inverse matrix operation and bisection methods are employed, with computational complexities of $\mathcal{O}(M^3)$ and $\mathcal{O}\left(\log_2(\frac{1}{\epsilon})\right)$, respectively, where $\epsilon$ represents the iteration accuracy. As such, the overall computational complexity for determining $\mathbf{W}$ is obtained as $\mathcal{O}\left(t_{1,\max}\log_2(\frac{1}{\epsilon})M^3\right)$, where $t_{1,\max}$ denotes the total number of iterations for solving problem \eqref{eq::subProblem11_FD}. Second, in step 4, the computational complexity primarily depends on the MO algorithm in Algorithm~\ref{alg_MO}, which has a complexity of $\mathcal{O}(M^2)$. Thus, the total complexity for determining the hybrid beamformer is given by $\mathcal{O}\left(t_{2,\max}M^2\right)$, where $t_{2,\max}$ denotes the total number of iterations for solving problem \eqref{eq::subProblem2_HB}.
Finally, in steps 5 to 7, the MM algorithm is adopted to optimize the positions of MAs. Specifically, in Algorithm~\ref{alg::MM}, in step 3, the complexities for calculating the maximum eigenvalues of $\mathbf{D}^U_m$ and $\mathbf{D}^U_m$ are given by $\mathcal{O}(L_U^3)$ and $\mathcal{O}(L_E^3)$, respectively. The complexity for calculating $\boldsymbol{\tau}_{m}^U, I \in \{U, E\}$ is obtained as $\mathcal{O}(L_U^2) + \mathcal{O}(L_E^2)$. The complexities for calculating $\nabla f_5(\mathbf{t}_{m})$, $\nabla^2 f_5(\mathbf{t}_{m})$, and $\delta_m$ are given by $\mathcal{O}\left(L_U+L_E\right)$, $\mathcal{O}\left(L_U+L_E\right)$, and $\mathcal{O}(1)$, respectively. Besides, updating $\mathbf{t}_m^{\langle t+1 \rangle}$ by solving problem \eqref{eq::subproblem2_transformed_final} incurs a complexity of $\mathcal{O}\left(M^{1.5}\ln(\frac{1}{\varepsilon})\right)$, where $\varepsilon$ represents the accuracy for the interior-point method. 
As such, the total complexity for obtaining $\mathbf{t}_m$ is given by $\mathcal{O}\left(L_U^3 + L_E^3 + t_{3,\max} \left(L_U^2 + L_E^2\right) + t_{4,\max}M^{1.5}\ln(\frac{1}{\varepsilon})\right)$, where $t_{1,\max}$ and $t_{2,\max}$ denote the maximum number of iterations for steps 2-8 and the maximum number of iterations for solving problem \eqref{eq::subproblem2_transformed_final}, respectively.
Therefore, the overall complexity of Algorithm~\ref{alg::overall} is obtained as $\mathcal{O}\Big(T_{\max}\big(t_{1,\max}\log_2(\frac{1}{\epsilon})M^3 + t_{2,\max}M^2 + L_U^3 + L_E^3 + t_{3,\max} \left(L_U^2 + L_E^2\right) + t_{4,\max}M^{1.5}\ln(\frac{1}{\varepsilon})\big)\Big)$, where $T_{\max}$ represents the total number of iterations for Algorithm~\ref{alg::overall}.

\section{Simulation Results}
In this section, we provide simulation results of the considered MA-enabled secure near-field MIMO systems. 

\subsection{Simulation Setup}
The user and eavesdropper are each equipped with uniform planar arrays (UPAs) based on FPA, having $L_U = 4$ and $L_E = 4$, respectively, while the BS is equipped with $M=64$ MAs. In a 3D Cartesian coordinate system, the BS is positioned on the $y$-$O$-$z$-plane, with its central point at the origin $(0, 0, 0)$. The user and eavesdropper are situated on the $x$-$O$-$y$-plane. For convenience, the positions of the BS, user, and eavesdropper are described in terms of polar coordinates. Specifically, this paper addresses a challenging case in which the eavesdropper lies between the BS and user, sharing the same azimuth angles. Unless stated otherwise, the user and eavesdropper are positioned at $(15 {\rm m}, \pi/4)$ and $(10 {\rm m}, \pi/4)$, respectively.
In addition, we assume equal noise variances across both the user and eavesdropper, i.e., $\sigma^2_U = \sigma^2_E = -80$ dBm. The moving region for the MA is modeled as a 2D square with size $\mathcal{C} = [-\frac{A}{2},\frac{A}{2}] \times [-\frac{A}{2},\frac{A}{2}]$, where $A = 100\lambda$ and $\lambda = 0.01$ m denotes the wavelength. Besides, the minimum spacing between MAs is $d_{\min} = \frac{\lambda}{2}$.
Unless otherwise stated, the following default simulation parameters are adopted. Specifically, $P_B = 20$ dBm, $N = 4$, $K = 2$, $\epsilon_1 = \epsilon_2 = \epsilon_3 = 10^{-6}$, $t_{\max}^{\rm MO} = t_{\max}^{\rm MM} = T_{\max} = 300$. 
All the results are averaged over 500 Monte Carlo simulations of independent channel realizations.

To comprehensively evaluate the superiority of the proposed scheme (denoted as ``\textbf{Proposed}''), we compared it against several benchmarks as outlined below.
\begin{itemize}
	\item \textbf{Fully-digital (FD):} The BS employs a fully-digital beamformer $\mathbf{W}$, which regards as a theoretical upper bound for the proposed approach.
	\item \textbf{Random position antenna (RPA):} The MAs at the BS are randomly distributed within the moving region $\mathcal{C}$.
	\item \textbf{FPA with full aperture (FPAF):} The BS employs a conventional FPA-based UPA, with antennas arranged to achieve the maximum possible aperture of $A \times A$.
	\item \textbf{FPA with half-wavelength antenna spacing (FPAH):} The BS uses a traditional FPA-based UPA, where the antennas are arranged with half-wavelength spacing.
%	 in both the horizontal and vertical directions.
	\item \textbf{Far-field (FF):} The channel between the BS and the user/eavesdropper is characterized by a traditional far-field channel model. Antenna position optimization is conducted according to the far-field channel model, whereas performance evaluation is carried out using the true near-field channel model, as defined in \eqref{eq::NFRV}.
\end{itemize}

\subsection{Convergence Evaluation of Algorithm~\ref{alg::overall}}

\begin{figure}[t]
	\begin{center}
		\includegraphics[width= 3.7 in]{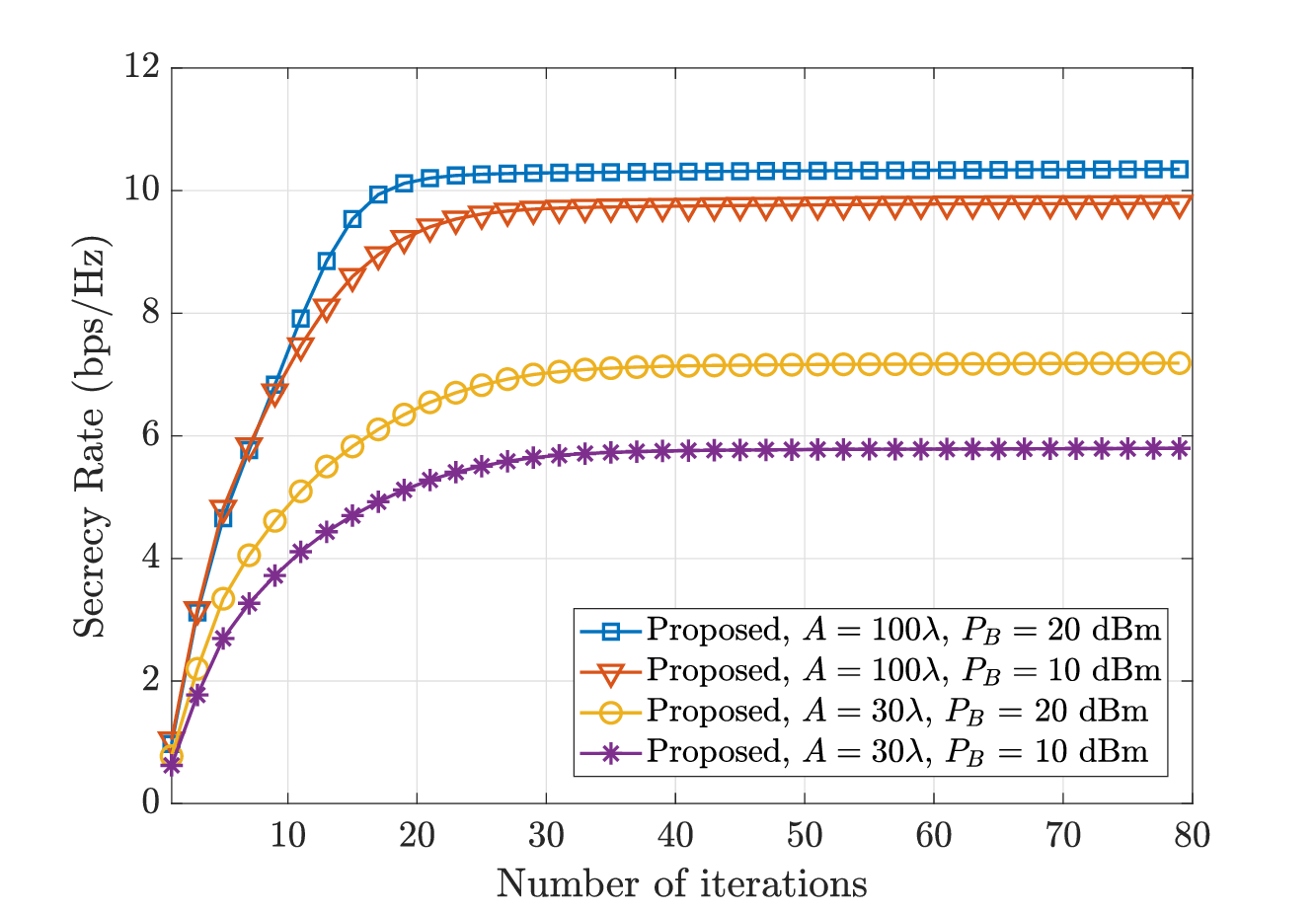}
		\caption{The convergence behavior of  Algorithm~\ref{alg::overall}.} \label{fig::convergence}
	\end{center}
	\vspace{-1.5em}
\end{figure}
Herein, the convergence performance of the proposed overall Algorithm~\ref{alg::overall} is evaluated in Fig. \ref{fig::convergence}. It is evident that regardless of the moving region size or the transmit power budget, the achievable secrecy rate consistently increases and stabilizes after approximately 30 iterations. Specifically, with $A = 100\lambda$ and $P_B = 20$ dBm, the secrecy rate improves from $1.042$ bps/Hz to $10.29$ bps/Hz, highlighting the effectiveness of the proposed solution in enhancing the PLS of the MA-aided near-field system. Moreover, it can be observed that the impact of adjusting the size of the moving region is greater than that of adjusting the transmit power. This is because insufficient transmit power can be compensated by leveraging the excess spatial DoFs, thereby improving the PLS of the system.

\subsection{Beam Focusing in MA-Aided Near-Field PLS}

%\begin{figure}[t]
%	\begin{center}
%		\includegraphics[width= 3.7 in]{figures//heatmap_FF.eps}
%		\caption{Normalized heat-maps for beam steering when $A = 100\lambda$.} \label{fig::heatmap_FF}
%	\end{center}
%	\vspace{-1.5em}
%\end{figure}
%
%\begin{figure}[t]
%	\begin{center}
%		\includegraphics[width= 3.68 in]{figures//heatmap_100lambda.eps}
%		\caption{Normalized heat-maps for beam focusing when $A = 100\lambda$.} \label{fig::heatmap_100lambda}
%	\end{center}
%	\vspace{-1.5em}
%\end{figure}
%
%\begin{figure}[t]
%	\begin{center}
%		\includegraphics[width= 3.7 in]{figures//heatmap_10lambda.eps}
%		\caption{Normalized heat-maps for beam focusing when $A = 30\lambda$.} \label{fig::heatmap_10lambda}
%	\end{center}
%	\vspace{-1.5em}
%\end{figure}

% 1*3 构图
\begin{figure*}[ht]
	\begin{center}
%			\begin{minipage}[b]{0.3\linewidth}
%					\centering
%					\includegraphics[width= 2.3 in]{figures//heatmap_FF.eps}
%					\caption{Normalized heat-maps for beam steering when $A = 100\lambda$.} \label{fig::heatmap_FF}
%				\end{minipage}
%			\quad 
			\begin{minipage}[b]{0.3\linewidth}
					\centering
					\includegraphics[width= 2.36 in]{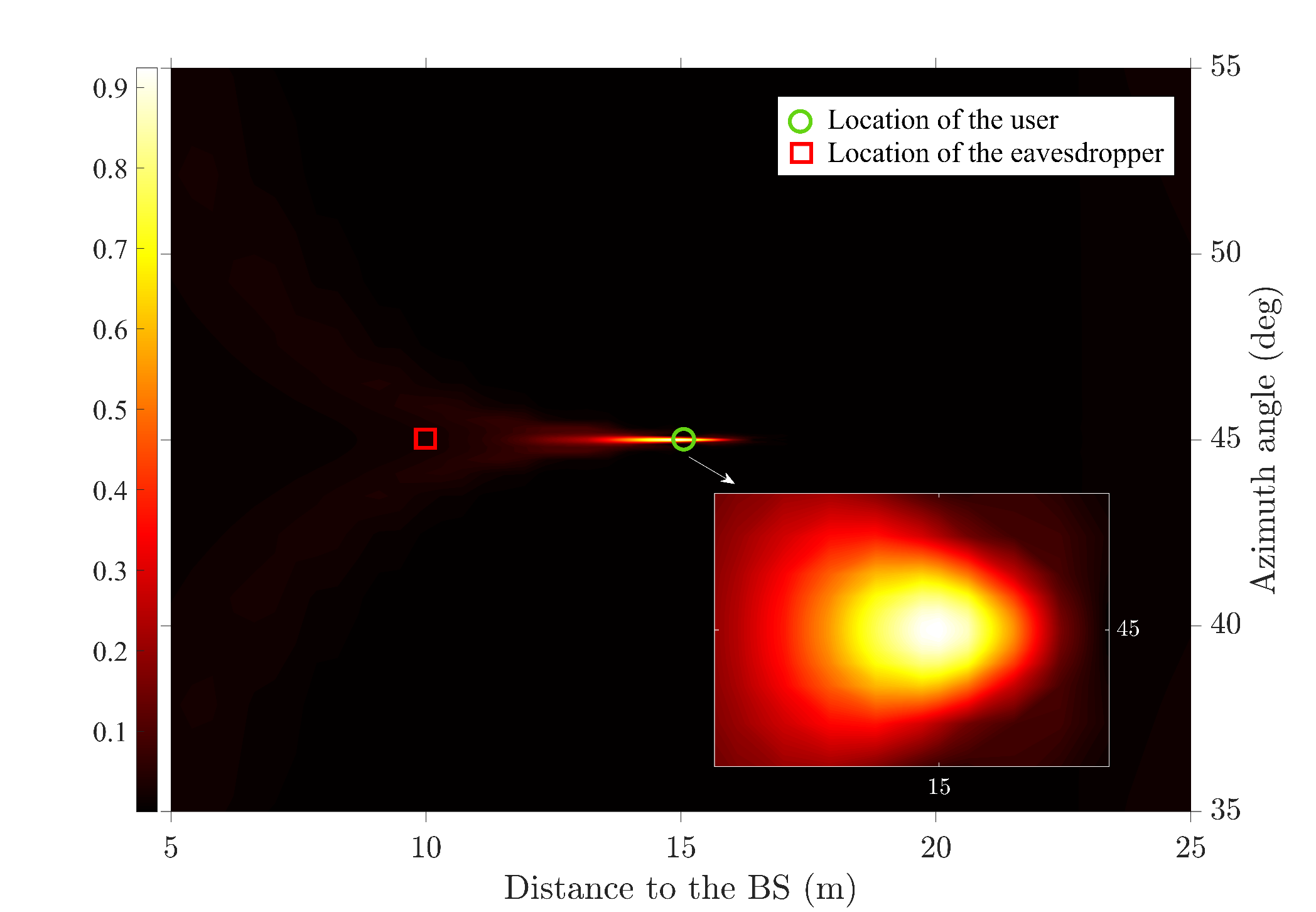}
					\caption{Normalized heat-maps for beam focusing when $A = 100\lambda$.} \label{fig::heatmap_100lambda}
				\end{minipage}
			\quad 
			\begin{minipage}[b]{0.3\linewidth}
					\centering
					\includegraphics[width= 2.4 in]{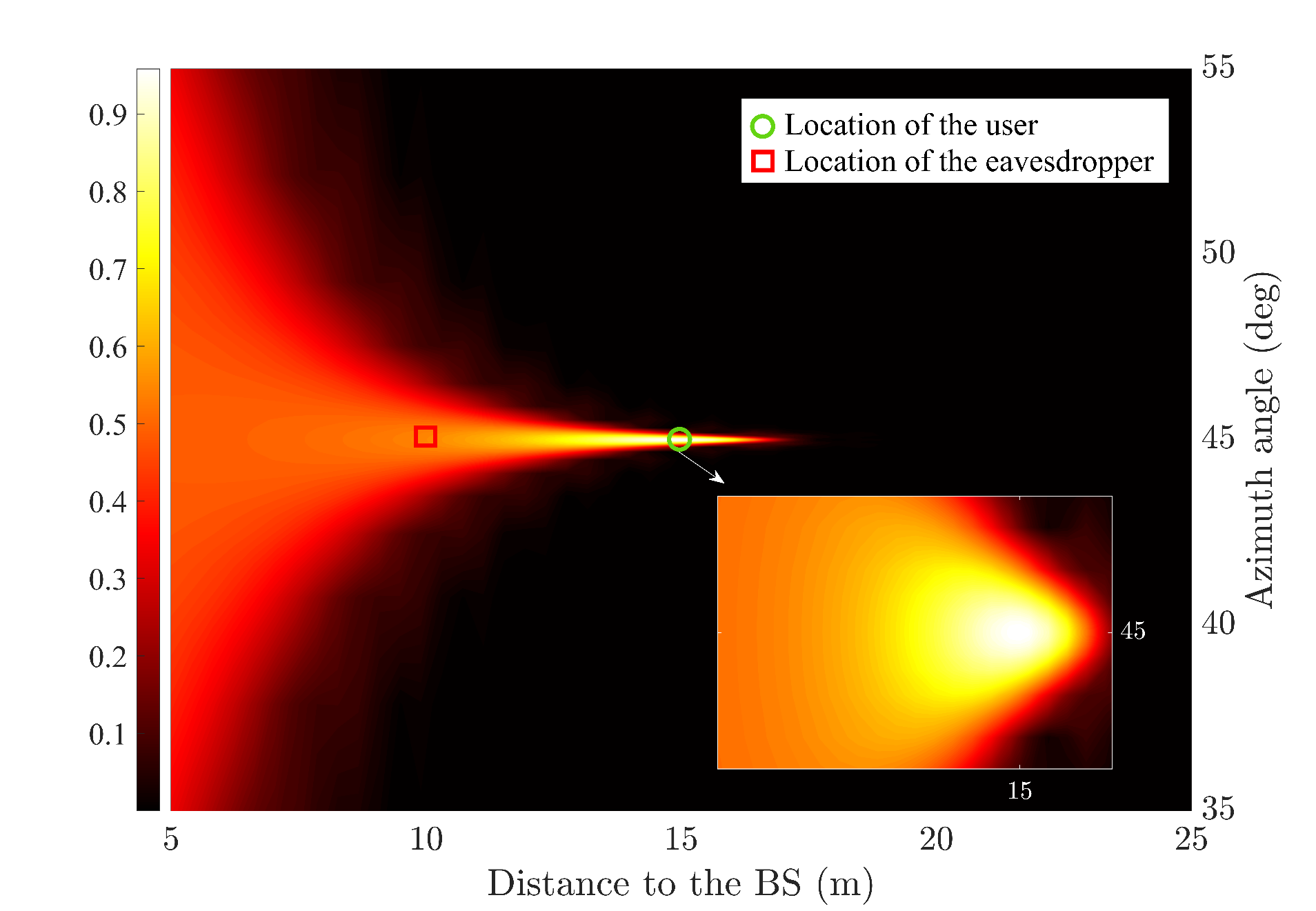}
					\caption{Normalized heat-maps for beam focusing when $A = 30\lambda$.} \label{fig::heatmap_30lambda}
				\end{minipage}
			\quad
			\begin{minipage}[b]{0.3\linewidth}
				\centering
				\includegraphics[width= 2.35 in]{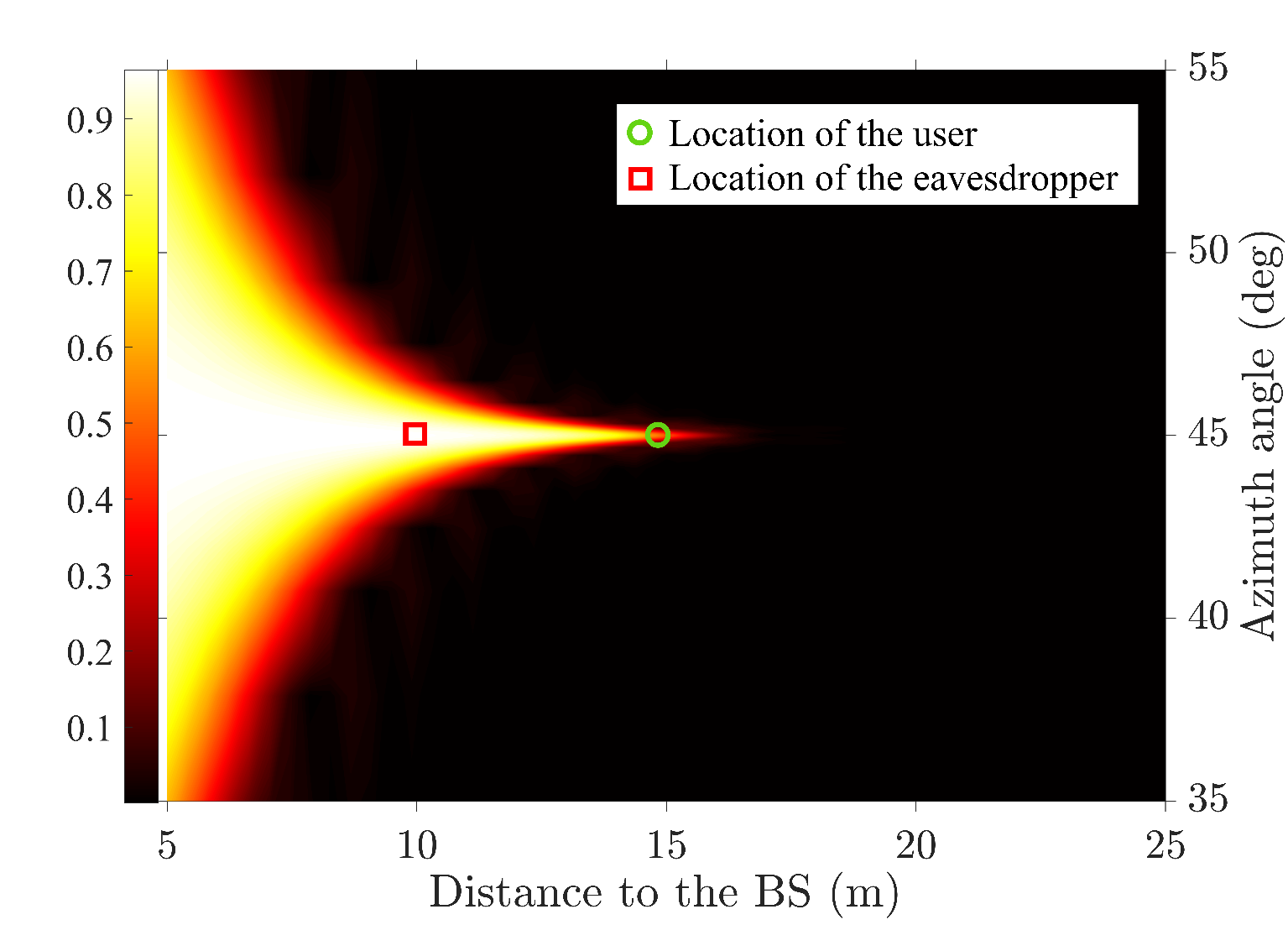}
				\caption{Normalized heat-maps for beam focusing when $A = 10\lambda$.} \label{fig::heatmap_10lambda}
			\end{minipage}
		\end{center}
	\vspace{-1.5em}
\end{figure*}

%Figs.~\ref{fig::heatmap_100lambda}-\ref{fig::heatmap_10lambda} demonstrate the normalized power heat-maps illustrated two distinct beam patterns. Specifically, Fig.~\ref{fig::heatmap_FF} displays the beam steering pattern in the considered MA-assisted far-field MIMO communication systems, where the eavesdropper always receives higher power than the legitimate user, preventing the achievement of a positive secrecy rate. As such, we can conclude that when the user and eavesdropper are situated in the same direction w.r.t. the BS, far-field systems fail to ensure secure communication between them.
Figs.~\ref{fig::heatmap_100lambda}-\ref{fig::heatmap_10lambda} illustrate the beam focusing patterns in the considered near-field scenarios, where positions of MAs are optimized via the proposed Algorithm~\ref{alg::MM}. As shown in Fig.~\ref{fig::heatmap_100lambda}, when $A = 100\lambda$, the beam is accurately concentrated on the intended legitimate user, and the beam exhibits an exceptionally narrow main lobe. This configuration minimizes information leakage to the eavesdropper and facilitates secure transmission through joint optimization of both the angle and distance domains. In addition, as demonstrated in Figs.~\ref{fig::heatmap_30lambda}-\ref{fig::heatmap_10lambda}, reducing the size of the moving region leads to an expansion of the beam's primary lobes, attributed to the reduction in the maximum aperture achievable by the MA array. This expansion results in increased information leakage to the eavesdropper at undesirable locations, thus compromising the overall PLS performance.

\subsection{Performance Analysis Compared to Baseline Schemes}

\begin{figure}[t]
	\begin{center}
		\includegraphics[width= 3.7 in]{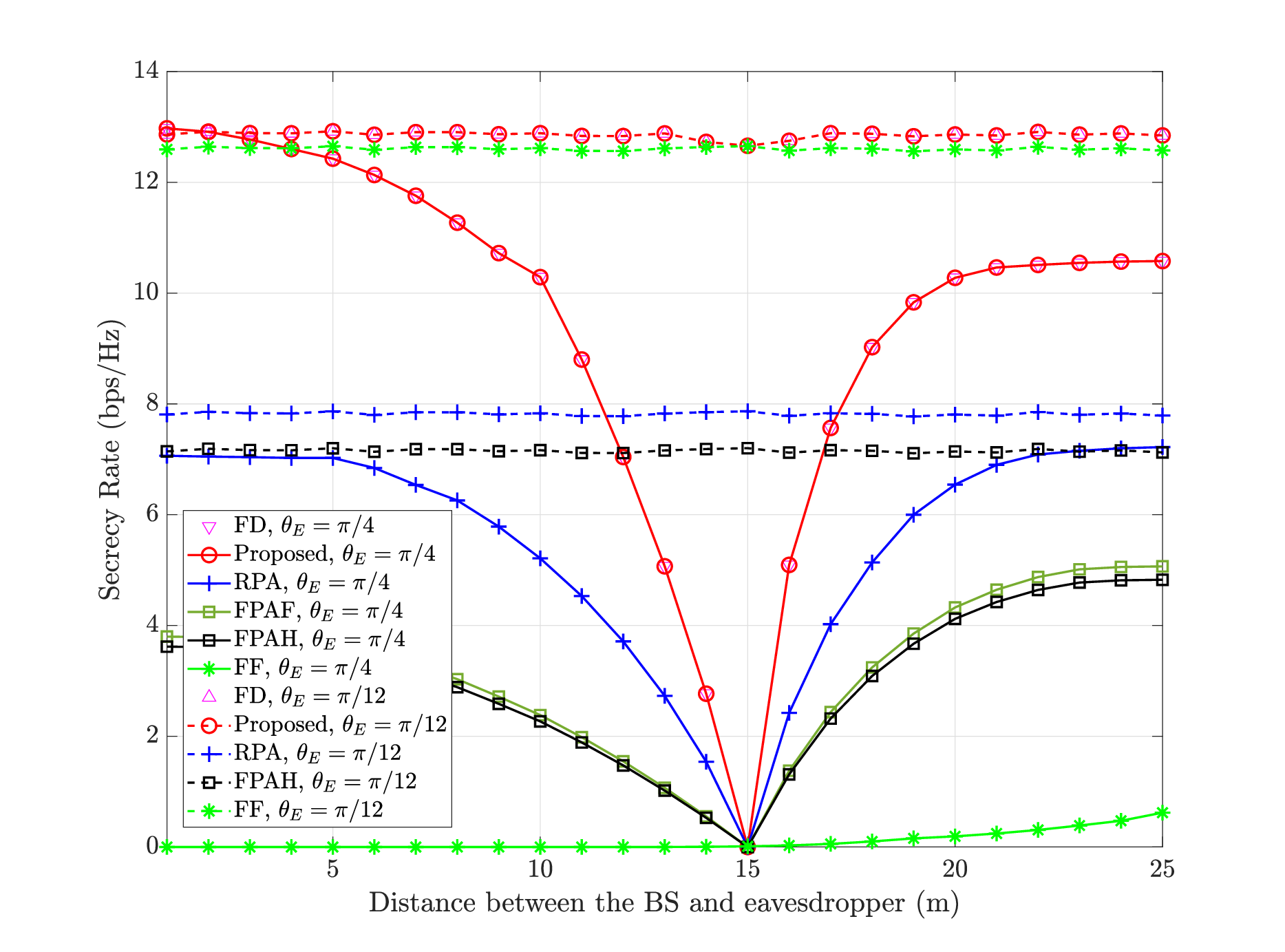}
		\caption{Secrecy rate versus distance between the BS and eavesdropper.} \label{fig::SR_vs_distance}
	\end{center}
	\vspace{-1.5em}
\end{figure}
In Fig.~\ref{fig::SR_vs_distance}, we compare the secrecy rate of the proposed scheme and baselines versus the distance between the BS and eavesdropper under different azimuth angles of the eavesdropper. 
When the eavesdropper is aligned with the user at the azimuth angle of $\pi/4$, the proposed MA-aided approach consistently outperforms RPA, FPAF, and FPAH schemes in terms of secure transmission performance. Specifically, the secrecy rate diminishes to zero as the eavesdropper’s position shifts from $(0 {\rm m}, \pi/4)$ to the user’s position $(15 {\rm m}, \pi/4)$, and subsequently increases as the location of the eavesdropper further moves further away the user. Furthermore, the optimized hybrid beamforming technique yields a secrecy rate that is comparable to that of fully-digital beamformers. However, for the FF scheme, the secrecy rate approaches zero when the eavesdropper is situated nearer to the BS than the user.
When the eavesdropper is located at the azimuth angle $\pi/12$, distinct from the user’s $\pi/4$, the secrecy rate remains largely unaffected by the distance between the BS and the eavesdropper for all schemes, including those in far-field environments.

\begin{figure}[t]
	\begin{center}
		\includegraphics[width= 3.7 in]{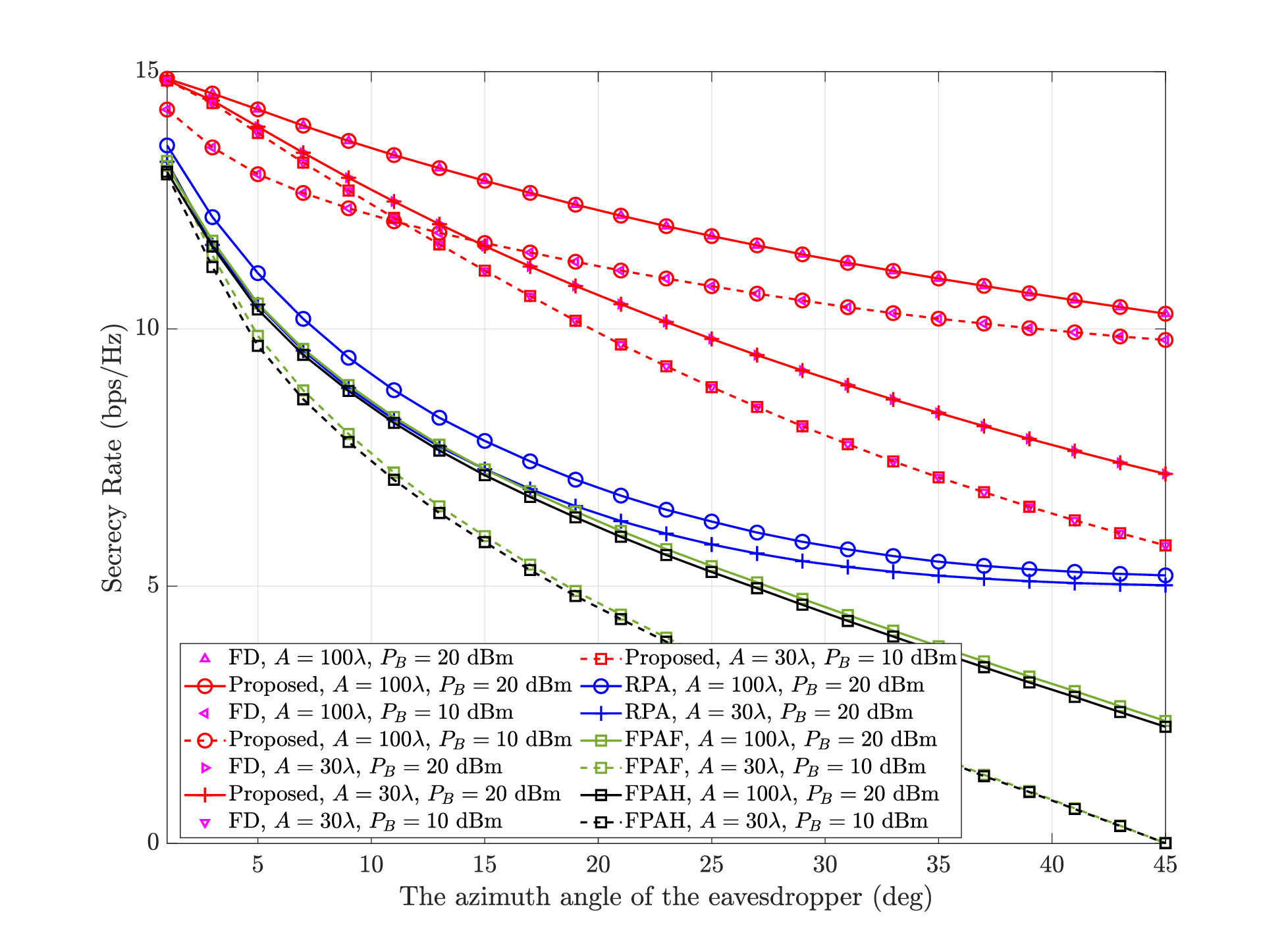}
		\caption{Secrecy rate versus the azimuth angle of the eavesdropper.} \label{fig::SR_vs_thetaE}
	\end{center}
	\vspace{-1.5em}
\end{figure}

In Fig.~\ref{fig::SR_vs_thetaE}, we evaluate the secrecy rate versus the azimuth angle of the eavesdropper, $\theta_E$, under different transmit powers and sizes of the moving region of MAs. It can be observed that the proposed scheme consistently achieves the highest secrecy rate across all configurations, particularly when the array aperture is large ($A=100\lambda$) and the transmit power is high ($P_B=20$ dBm). As $\theta_E$ increases from $0^{\circ}$ to $45^{\circ}$, the secrecy rate decreases for all schemes due to the reduced angular separation between the legitimate user and the eavesdropper, which leads to increased information leakage. Moreover, the impact of the array aperture is significant, as a larger aperture provides a sharper main lobe and better spatial resolution, thereby enhancing the ability to suppress eavesdropping. For instance, the proposed scheme with $A=100\lambda$ outperforms the scheme with $A=30\lambda$ by a notable margin across the entire angular range. Similarly, increasing $P_B$ improves the secrecy rate of the proposed method, reflecting a better signal-to-noise ratio (SNR) at the user while maintaining limited exposure to the eavesdropper. In contrast, the RPA and FPA-based methods show relatively lower secrecy performance, with the FPA-based scheme being the worst, especially for small apertures and lower transmit power. This highlights the effectiveness of the proposed solution in leveraging large-aperture arrays and high-power transmission to enhance PLS.

\begin{figure}[t]
	\begin{center}
		\includegraphics[width= 3.8 in]{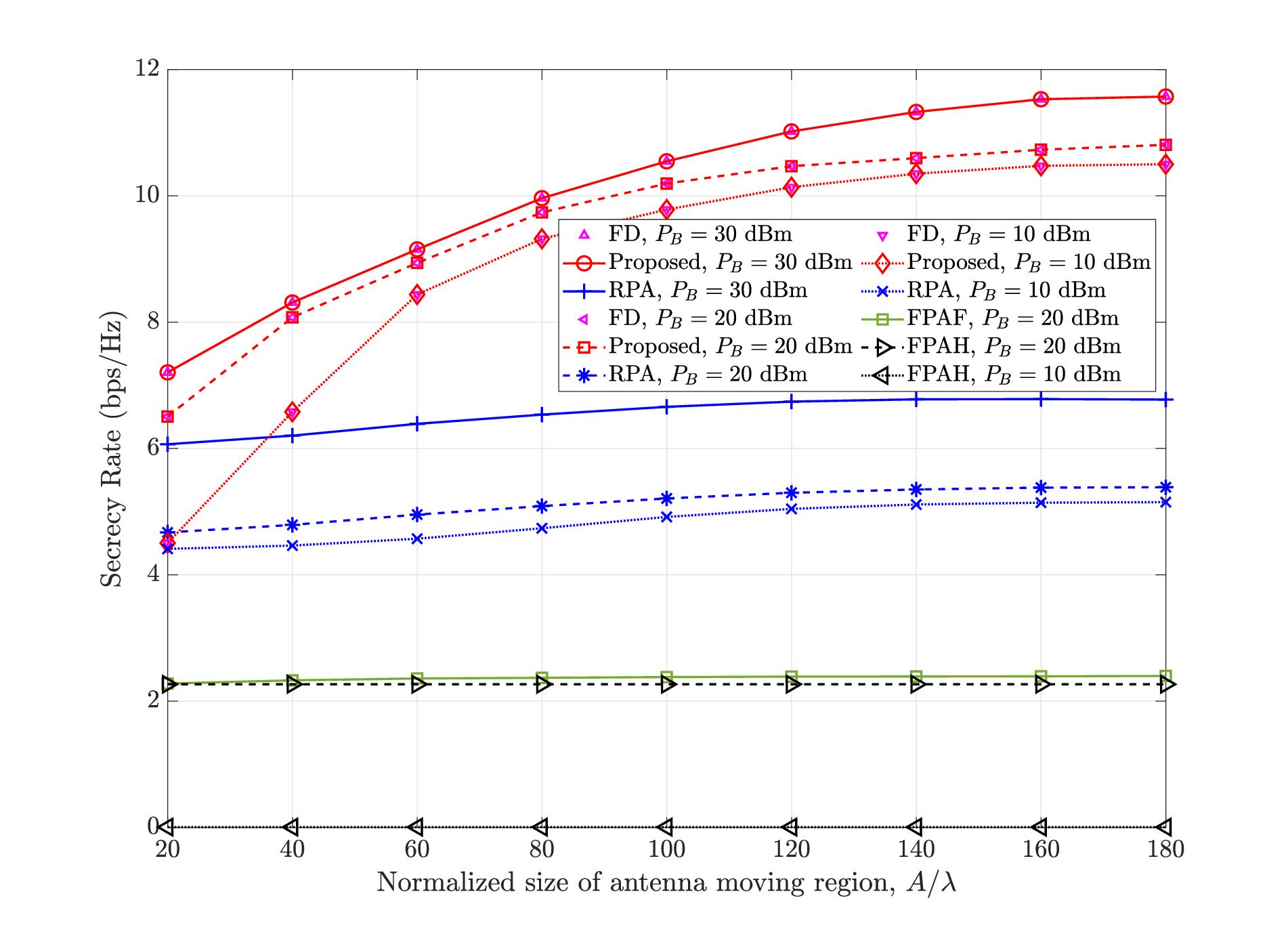}
		\caption{Secrecy rate versus the normalized size of the antenna moving region $A/\lambda$.} \label{fig::SR_vs_A}
	\end{center}
	\vspace{-1.5em}
\end{figure}

Fig.~\ref{fig::SR_vs_A} compares the secrecy rate versus the normalized size of the antenna moving region, i.e., $A/\lambda$, under different transmit powers $P_B$. The proposed scheme consistently outperforms both the RPA and FPA-based methods across all values of $A/\lambda$, especially for higher transmit powers. As $A/\lambda$ increases, the secrecy rate of the proposed, RPA, and FPAF scheme increases. This is because enlarging the moving region provides two key benefits. On the one hand, it expands the equivalent array aperture, thereby extending the near-field region where both the legitimate user and the eavesdropper can be spatially resolved. The proposed method fully exploits both advantages, thereby enhancing secrecy performance. In contrast, the FPAF approach benefits only from the latter, while the FPAH approach fails to capitalize on either benefit.
Thus, the secrecy rate for the proposed scheme with $P_B = 30$ dBm increases more steeply than in the other benchmarks, highlighting the positive effect of higher transmit power on secrecy rate. For the RPA scheme, the secrecy rate remains relatively stable as $A/\lambda$ increases and is consistently lower than that of the proposed scheme, especially for lower transmit power levels. In contrast, the FPAH method demonstrates the worst performance across all conditions, with its secrecy rate remaining nearly constant and significantly lower than the proposed and RPA schemes. 
%This trend is especially noticeable at lower transmit powers, where the FPA's ability to increase secrecy rate is limited.

\begin{figure}[t]
	\begin{center}
		\includegraphics[width= 3.7 in]{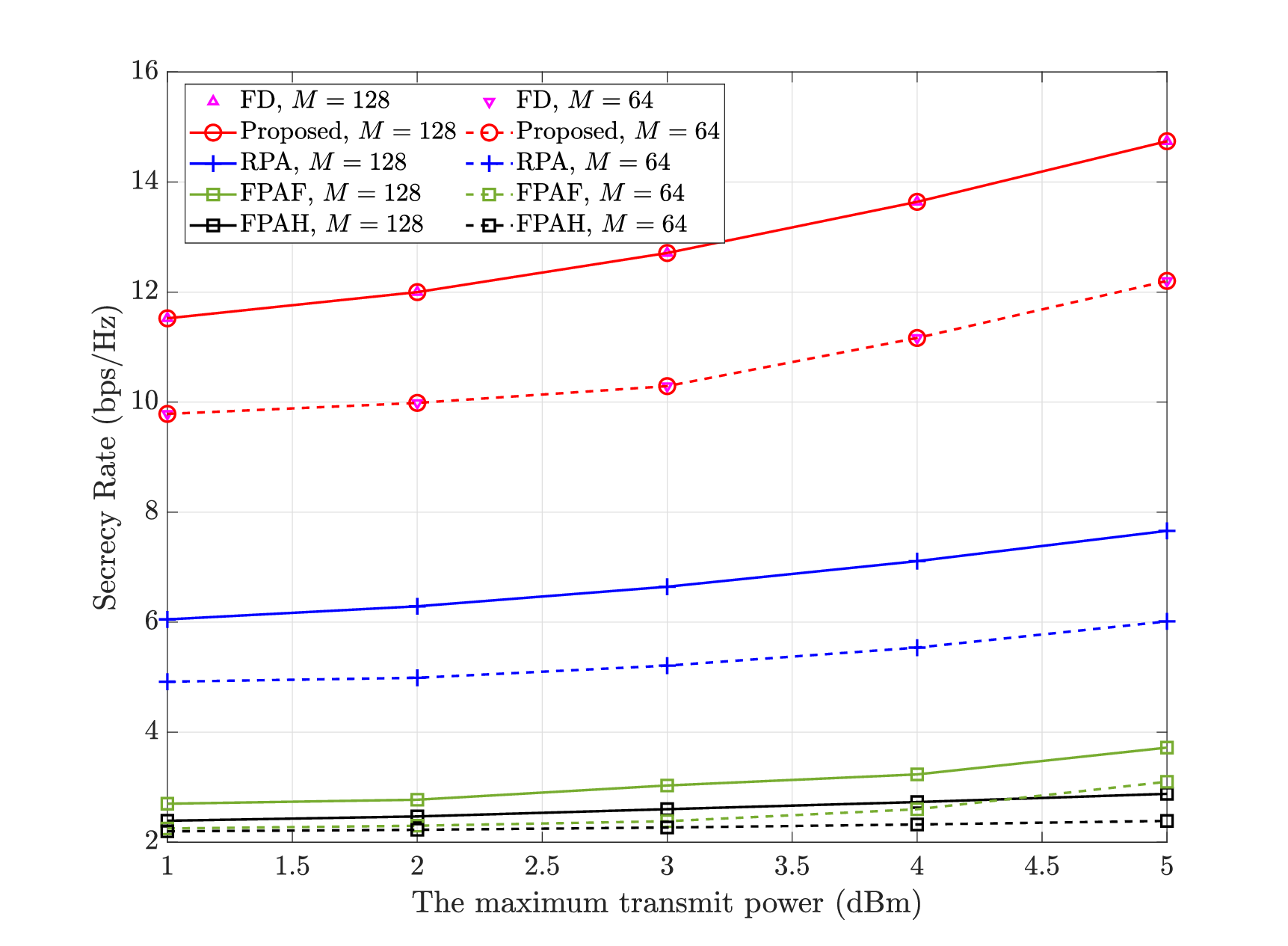}
		\caption{Secrecy rate versus transmit power budget $P_B$.} \label{fig::SR_vs_P}
	\end{center}
	\vspace{-1.5em}
\end{figure}
Fig.~\ref{fig::SR_vs_P} shows the variation in secrecy rate with the maximum transmit power for all schemes with different numbers of MAs $M$.
As shown in the figure, the proposed scheme consistently achieves the highest secrecy rate compared to both RPA and FPA-based schemes across all values of $M$. The secrecy rate enhances as the transmit power grows for all schemes, with a more noticeable improvement observed for $M=128$ compared to $M=64$. This demonstrates that a larger antenna array size enhances the system's ability to exploit the transmit power and enhance secrecy rate.
Specifically, for $M=128$, for $P_{\max} = 20$ dBm, the secrecy rate is improved by 147.7\% compared to the FPAH benchmark. Furthermore, in contrast to FPA-based methods, the proposed MA-assisted approach can reduce transmit power while maintaining the same secrecy rate by using antenna movement to efficiently adjust the channels.

\begin{figure}[t]
	\begin{center}
		\includegraphics[width= 3.7 in]{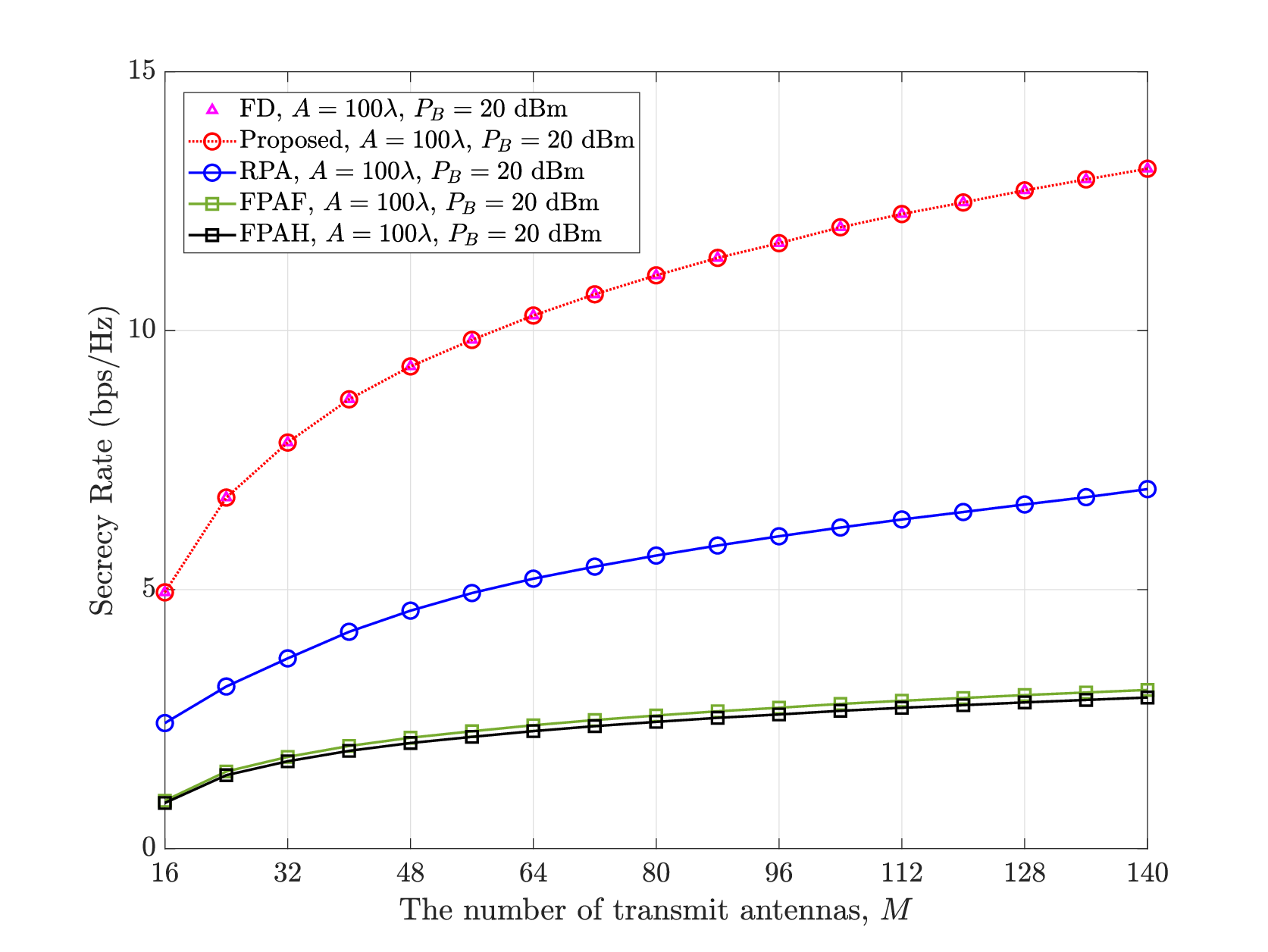}
		\caption{Secrecy rate versus the number of transmit antenna $M$.}\label{fig::SR_vs_M}
	\end{center}
	\vspace{-1.5em}
\end{figure}

Finally, Fig.~\ref{fig::SR_vs_M} shows the secrecy rate performance against the number of antennas at the BS. The secrecy rate provided by all schemes rises with the increase of $M$, suggesting that a broader antenna aperture contributes to improving secure communication performance. The proposed approach demonstrates a notable  performance enhancement compared to the baselines. Specifically, when $N = 140$, compared to the FPAH scheme, the secrecy rate reaches up to $13.13$ bps/Hz.
In addition, as $M$ grows larger, the disparity in performance between the proposed scheme and the baselines becomes increasingly evident, suggesting that a larger antenna array is crucial for maximizing secrecy rate. Furthermore, although increasing the number of antennas improves secrecy rate, the effect is limited.

% 1*3 构图
%\begin{figure*}[ht]
%	\begin{center}
%		\begin{minipage}[b]{0.3\linewidth}
%			\centering
%			\includegraphics[width= 2.45 in]{figures//SR_vs_Gammar.eps}
%			\caption{Sum-rate versus the radar SNR threshold $\Gamma_r$.} \label{fig::SR_vs_Gammar}
%		\end{minipage}
%		\quad 
%		\begin{minipage}[b]{0.3\linewidth}
%			\centering
%			\includegraphics[width= 2.4 in]{figures//SR_vs_Gamma_v2.eps}
%			\caption{Sum-rate versus the communication SINR threshold $\Gamma$.} \label{fig::SR_vs_Gamma}
%		\end{minipage}
%		\quad 
%		\begin{minipage}[b]{0.3\linewidth}
%			\centering
%			\includegraphics[width= 2.25 in]{figures//SR_vs_Gammae_v2.eps}
%			\caption{Sum-rate versus the eavesdropping SINR threshold $\Gamma_e$.} \label{fig::SR_vs_Gammae}
%		\end{minipage}
%	\end{center}
%	\vspace{-1.5em}
%\end{figure*}

\section{Conclusions}
In this paper, an MA-assisted secure transmission scheme for near-field MIMO communication systems was proposed, where an eavesdropper attempts to intercept secrecy data transmitted from the BS to the user. We formulated an optimization problem with the goal of maximizing the secrecy rate by jointly optimizing the hybrid beamformers and MA positions. To address the highly non-convex problem, we proposed an AO-based algorithm. Specifically, for the hybrid beamformers design subproblem, we first derived a semi-closed-form expression for the fully-digital beamformer, and subsequently, the hybrid beamformers at the BS were determined through MO techniques. Then, the MM algorithm was employed to solve the MA positions design subproblem. Simulation results confirm that the integration of MAs can significantly improve the security capabilities in near-field MIMO systems. In addition, the proposed scheme can realize secure transmission even when the eavesdropper is situated in the same direction as the user and positioned closer to the BS. Moreover, it was shown that enlarging the size of the MA moving region, adding more transmit power, and increasing the number of transmit antennas demonstrated a further boost to the secrecy rate of MA-aided systems.

\appendices

\section{Proof of Theorem 1} \label{app0}
Let us consider the matrix function as follows
\begin{equation} \small \label{eq::matrixfunction}
	\mathbf{E}(\mathbf{P},\mathbf{W}) \triangleq (\mathbf{I} - \mathbf{P}^{\rm H}\mathbf{H}\mathbf{W})(\mathbf{I} - \mathbf{P}^{\rm H}\mathbf{H}\mathbf{W})^{\rm H} + \mathbf{P}^{\rm H}\mathbf{P}.
\end{equation}
By checking the first-order optimality condition for the function $\log\det(\mathbf{Q}) - {\rm Tr}(\mathbf{Q}\mathbf{E})$, for any positive definite matrix $\mathbf{E} \in \mathbb{C}^{a \times a}$, we can have
\begin{equation} \small \label{eq::fact1}
	-\log\det(\mathbf{E}) = \max_{\mathbf{Q} \succ 0} \log\det(\mathbf{Q}) - {\rm Tr}(\mathbf{Q}\mathbf{E}) + a,
\end{equation}
where the optimal solution is given by $\mathbf{Q}^{\star} = \mathbf{E}^{-1}$. In addition, for any positive definite matrix $\mathbf{W}$, we have
\begin{equation} \small \label{eq::optimalP}
	\mathbf{P}^{\star} \!\triangleq\! (\mathbf{I} + \mathbf{H}\mathbf{W}\mathbf{W}^{\rm H}\mathbf{H}^{\rm H})^{-1}\mathbf{H}\mathbf{W} \!=\! \arg\min_{\mathbf{P}} {\rm Tr}(\mathbf{Q}\mathbf{E}(\mathbf{P}, \mathbf{W})).
\end{equation}
Substituting the optimal $\mathbf{P}^{\star}$ into $\mathbf{E}(\mathbf{P},\mathbf{W})$, we have
\begin{equation} \small \label{eq::fact2}
	\mathbf{E}(\mathbf{P}^{\star}, \mathbf{W}) = (\mathbf{I} + \mathbf{W}^{\rm H}\mathbf{H}^{\rm H}\mathbf{H}\mathbf{W})^{-1}.
\end{equation}
Based on the results from \eqref{eq::fact1} and \eqref{eq::fact2}, we have 
\begin{equation} \small \label{eq::fact3}
	\begin{aligned}
	 &\max_{\mathbf{Q} \succ 0, \mathbf{P}} \log\det(\mathbf{Q}) - {\rm Tr}(\mathbf{Q}\mathbf{E}) + a\\ &\!=\! \log\det(\mathbf{I} + \mathbf{W}^{\rm H}\mathbf{H}^{\rm H}\mathbf{H}\mathbf{W})  \overset{(\text{a})}{=} \log\det(\mathbf{I} + \mathbf{H}\mathbf{W}\mathbf{W}^{\rm H}\mathbf{H}^{\rm H}),
	\end{aligned} 
\end{equation}
where (a) follows the identity $\log\det(\mathbf{I} + \mathbf{X}\mathbf{Y}) = \log\det(\mathbf{I} + \mathbf{Y}\mathbf{X})$.

Therefore, by substituting $\mathbf{W} = \mathbf{W}_D$ and  $\mathbf{H} = \tilde{\mathbf{H}}(\mathbf{t})$ into \eqref{eq::fact3}, we have
\begin{equation} \small \label{eq::subProblem1_FD_transform_U}
	\begin{aligned}
		&\log_2\det\left(\mathbf{I}_{L_U} + \tilde{\mathbf{H}}(\mathbf{t})\mathbf{W}_D\mathbf{W}_D^{\rm H}\tilde{\mathbf{H}}^{\rm H}(\mathbf{t})\right) \\
		&\quad\quad\quad = \max_{\mathbf{Q}_U \succ 0, \mathbf{P}} \log\det(\mathbf{Q}_U) - {\rm Tr}(\mathbf{Q}_U\mathbf{E}) + a.
	\end{aligned}
\end{equation}
Furthermore, by substituting $\mathbf{H} = \tilde{\mathbf{Z}}(\mathbf{t})$ into \eqref{eq::fact3} and let $\mathbf{E} = \mathbf{I}_{L_E} + \tilde{\mathbf{Z}}(\mathbf{t})\mathbf{W}_D\mathbf{W}_D^{\rm H}\tilde{\mathbf{Z}}^{\rm H}(\mathbf{t})$, we obtain
\begin{equation} \small \label{eq::subProblem1_FD_transform_E}
	\begin{aligned}
		&- \log_2\det\left(\mathbf{I}_{L_E} + \tilde{\mathbf{Z}}(\mathbf{t})\mathbf{W}_D\mathbf{W}_D^{\rm H}\tilde{\mathbf{Z}}^{\rm H}(\mathbf{t})\right)\\
		&\max_{\mathbf{Q}_E \succ 0} \log\det(\mathbf{Q}_E) - {\rm Tr}(\mathbf{Q}_E(\mathbf{I}_{L_E} + \tilde{\mathbf{Z}}(\mathbf{t})\mathbf{W}_D\mathbf{W}_D^{\rm H}\tilde{\mathbf{Z}}^{\rm H}(\mathbf{t}))) + a,
	\end{aligned}
\end{equation}
which is equivalent to the sum of the right-hand-side of \eqref{eq::subProblem1_FD_transform_U} and \eqref{eq::subProblem1_FD_transform_E}, this completes the proof.

\section{Derivation of $\nabla f_5(\mathbf{t}_{m}^{\langle t \rangle})$ and $\nabla^2 f_5(\mathbf{t}_{m}^{\langle t \rangle})$} \label{app1}
Denote $\tau_{m,i}^U$ as the $i$-th element of $\boldsymbol{\tau}_{m}^U$, and $\tau_{m,j}^E$ as the $j$-th element of $\boldsymbol{\tau}_{m}^E$, and thus \eqref{eq::subproblem2_transformed4} can be represented by
\begin{equation} \small
	f_5(\mathbf{t}_{m}) \!=\! 2\left(\sum_{i=1}^{L_U}\varrho^U_{m,i}\cos\left(\varphi_{m,i}^U\right) \!+\! \sum_{j=1}^{L_E}\varrho^E_{m,j}\cos\left(\varphi_{m,j}^E\right)\right), 
\end{equation}
where $\varrho^U_{m,i} \triangleq \left|\tau_{m,i}^U\right|\left|g^U_{m,i}\right|$, $\varrho^E_{m,j} \triangleq \left|\tau_{m,j}^E\right|\left|g^E_{m,j}\right|$, $\varphi_{m,i}^U(\mathbf{t}_{m}) = \frac{2\pi}{\lambda}\|\mathbf{t}_{m}-\mathbf{r}^U_i\|_2 - \angle \varrho^U_{m,i}$, and $\varphi_{m,j}^E(\mathbf{t}_{m}) = \frac{2\pi}{\lambda}\|\mathbf{t}_{m}-\mathbf{r}^E_j\|_2 - \angle \varrho^E_{m,j}$, respectively. 
Based on the Fresnel approximation, $\|\mathbf{t}_{m}-\mathbf{r}^I_l\|_2, I \in \{U,E\}$, can be approximated as
\begin{equation} \small
	\begin{aligned}
		&\|\mathbf{t}_{m}-\mathbf{r}^I_l\| \approx r^I_l - y_{m}\sin\theta^I_l \sin\phi^I_l - z_{m}\cos\phi^I_l\\
		&\quad\quad\quad + \frac{y^2_{m} + z^2_{m} - \left(y_{m}\sin\theta^I_l \sin\phi^I_l + z_{m}\cos\phi^I_l\right)^2}{2 r^I_l}\\ &\quad\quad\quad  \triangleq \gamma^I(\mathbf{t}_{m}), I \in \{U,E\}.
	\end{aligned}
\end{equation}
Then, the gradient vector and the Hessian matrix of $f_5(\mathbf{t}_{m})$ w.r.t. $\mathbf{t}_{m}$ is obtained as $\nabla f_5(\mathbf{t}_{m}) = \left[\frac{\partial f_5(\mathbf{t}_{m})}{\partial y_{m}}, \frac{\partial f_5(\mathbf{t}_{m})}{\partial z_{m}}\right]^{\rm T}$ and $\nabla^2 f_5(\mathbf{t}_{m}) = \bigg[\begin{array}{cc}
	\frac{\partial^2 f_5(\mathbf{t}_{m})}{\partial y_{m}\partial y_{m}} & \frac{\partial^2 f_5(\mathbf{t}_{m})}{\partial y_{m}\partial z_{m}} \\
	\frac{\partial^2 f_5(\mathbf{t}_{m})}{\partial z_{m}\partial y_{m}} & \frac{\partial^2 f_5(\mathbf{t}_{m})}{\partial z_{m}\partial z_{m}} 
\end{array}\bigg]$, and the relative terms are shown at the top of the next page. 

\begin{figure*} [ht] 
	\centering
%	\hrulefill 
	\vspace*{8pt} 
	\begin{subequations} \small \label{eq::first-order}
		\begin{align}
			\frac{\partial f_5(\mathbf{t}_{m})}{\partial y_{m}} &= -\frac{4 \pi}{\lambda} \left(\sum_{i=1}^{L_U}\varrho^U_{m,i}\sin\left(\varphi_{m,i}^U\right)\frac{\partial \gamma^U(\mathbf{t}_{m})}{\partial y_{m}} + \sum_{j=1}^{L_E}\varrho^E_{m,j}\sin\left(\varphi_{m,j}^E\right)\frac{\partial \gamma^E(\mathbf{t}_{m})}{\partial y_{m}}\right),\\
			 \frac{\partial \gamma^I(\mathbf{t}_{m})}{\partial y_{m}} &= \left(\frac{y_{m} - \left(y_{m}\sin\theta^I_l \sin\phi^I_l + z_{m}\cos\phi^I_l\right)\sin\theta^I_l \sin\phi^I_l}{r^I_l}- \sin\theta^i_l \sin\phi^I_l\right), I \in \{U,E\}, \\
			 \frac{\partial f_5(\mathbf{t}_{m})}{\partial z_{m}} &= -\frac{4 \pi}{\lambda} \left(\sum_{i=1}^{L_U}\varrho^U_{m,i}\sin\left(\varphi_{m,i}^U\right)\frac{\partial \gamma^U(\mathbf{t}_{m})}{\partial z_{m}} + \sum_{j=1}^{L_E}\varrho^E_{m,j}\sin\left(\varphi_{m,j}^E\right)\frac{\partial \gamma^E(\mathbf{t}_{m})}{\partial z_{m}}\right),\\ 
			 \frac{\partial \gamma^I(\mathbf{t}_{m})}{\partial z_{m}} &= \left(\frac{z_{m} - \left(y_{m}\sin\theta^I_l \sin\phi^I_l + z_{m}\cos\phi^I_l\right) \cos\phi^I_l}{r^I_l}- \cos\phi^I_l\right), I \in \{U,E\}.
		\end{align}
\end{subequations}
\vspace{-2.5 em} 	
\end{figure*}
		 
\begin{figure*} [ht] 
	\centering
	%	\hrulefill 
	\vspace*{8pt} 
	\begin{subequations} \small \label{eq::second-order}
		\begin{align}			 
			\frac{\partial^2 f_5(\mathbf{t}_{m})}{\partial y_{m}\partial y_{m}} &= -\frac{8 \pi^2}{\lambda^2} \left(\sum_{i=1}^{L_U}(\varrho^U_{m,i})^2\cos\left(\varphi_{m,i}^U\right)\left(\frac{\partial \gamma^U(\mathbf{t}_{m})}{\partial y_{m}}\right)^2 + \sum_{j=1}^{L_E}(\varrho^E_{m,j})^2\cos\left(\varphi_{m,j}^E\right)\left(\frac{\partial \gamma^E(\mathbf{t}_{m})}{\partial y_{m}}\right)^2\right),\\
			&\!\!\!\!\!\!\!\!\!\!\!\!\!\!\!\!\!\!\!\!\!\!\!\!\!\!\frac{\partial^2 f_5(\mathbf{t}_{m})}{\partial y_{m}\partial z_{m}} =\frac{\partial^2 f(\mathbf{t}_{m})}{\partial z_{m}\partial y_{m}} \nonumber \\
			&= -\frac{8 \pi^2}{\lambda^2} \left(\sum_{i=1}^{L_U}(\varrho^U_{m,i})^2\cos\left(\varphi_{m,i}^U\right)\frac{\partial \gamma^U(\mathbf{t}_{m})}{\partial y_{m}}\frac{\partial \gamma^U(\mathbf{t}_{m})}{\partial z_{m}} + \sum_{j=1}^{L_E}(\varrho^E_{m,j})^2\cos\left(\varphi_{m,j}^E\right)\frac{\partial \gamma^E(\mathbf{t}_{m})}{\partial y_{m}}\frac{\partial \gamma^E(\mathbf{t}_{m})}{\partial z_{m}}\right),\\
			 \frac{\partial^2 f_5(\mathbf{t}_{m})}{\partial z_{m}\partial z_{m}} &= -\frac{8 \pi^2}{\lambda^2} \left(\sum_{i=1}^{L_U}(\varrho^U_{m,i})^2\cos\left(\varphi_{m,i}^U\right)\left(\frac{\partial \gamma^U(\mathbf{t}_{m})}{\partial z_{m}}\right)^2 + \sum_{j=1}^{L_E}(\varrho^E_{m,j})^2\cos\left(\varphi_{m,j}^E\right)\left(\frac{\partial \gamma^E(\mathbf{t}_{m})}{\partial z_{m}}\right)^2\right).
		\end{align}
	\end{subequations}
		\hrulefill
		\vspace{-1 em}
\end{figure*}

%\section{Construction of $\delta_k$} \label{app2}
%Since we have
%\begin{equation} \small
%	\begin{aligned}
%		\left\|\nabla^2\hat{\psi}_k(\mathbf{u}_k)\right\|_2^2 & \!\leq\! 	\left\|\nabla^2\hat{\psi}_k(\mathbf{u}_k)\right\|_F^2  \!\leq\! 4\Big(\frac{8\pi^2}{\lambda^2}\sum_{i=1}^{L_k^r}\left|\varsigma_{k,i}^{(t)}\right|\Big)^2,
%	\end{aligned}
%\end{equation}
%and
%\begin{equation} \small
%	\left\|\nabla^2\hat{\psi}_k(\mathbf{u}_k)\right\|_2\mathbf{I}_2 \succeq \nabla^2\hat{\psi}_k(\mathbf{u}_k),
%\end{equation}
%thus we can select $\delta_k$ as
%\begin{equation} \label{eq::delta} \small
%	\delta_k = \frac{16\pi^2}{\lambda^2}\sum_{i=1}^{L_k^r}\left|\varsigma_{k,i}^{(t)}\right|,
%\end{equation}
%which is satisfied the following condition
%\begin{equation} \small
%	\delta_k\mathbf{I}_2 \succeq \left\|\nabla^2\hat{\psi}_k(\mathbf{u}_k)\right\|_2\mathbf{I}_2 \succeq \nabla^2\hat{\psi}_k(\mathbf{u}_k).
%\end{equation}

\bibliographystyle{IEEEtran}
\bibliography{./IEEEabrv,./MANF}

% that's all folks%
\end{document}